\documentclass[12pt, draftclsnofoot, onecolumn]{IEEEtran}

\usepackage[noadjust]{cite}

\ifCLASSINFOpdf
   \usepackage[pdftex]{graphicx}
   \graphicspath{{../pdf/}{../jpeg/}}
   \DeclareGraphicsExtensions{.pdf,.jpeg,.png}
\else
   \usepackage[dvips]{graphicx}
   \graphicspath{{../eps/}}
   \DeclareGraphicsExtensions{.eps}
\fi

\usepackage[font=footnotesize,labelformat=simple]{subcaption}

\usepackage{amsmath}
\usepackage{amssymb}
\usepackage{bm}
\usepackage{amsfonts}

\usepackage{epstopdf}
\newcommand\numberthis{\addtocounter{equation}{1}\tag{\theequation}}
\usepackage[normalem]{ulem}
\usepackage[binary-units = true]{siunitx}

\usepackage{tikz}
\usepackage{pgfplots}
\usetikzlibrary{shapes, shapes.geometric, arrows}

\usepackage{amsthm}
\newtheoremstyle{exampstyle}
{3pt} 
{0pt} 
{} 
{} 
{\bfseries} 
{.} 
{.5em} 
{} 
\theoremstyle{exampstyle} 
\newtheorem{theorem}{Theorem}
\theoremstyle{exampstyle} 
\newtheorem{lemma}{Lemma}
\theoremstyle{exampstyle} 
\newtheorem{corollary}{Corollary}
\theoremstyle{exampstyle} 
\newtheorem{problem}{Problem}
\theoremstyle{exampstyle}

\usepackage[ruled,linesnumbered]{algorithm2e}

\usepackage[font=footnotesize,labelsep=period]{caption}

\begin{document}

\title{Integrated Access and Backhaul Optimization for Millimeter Wave Heterogeneous Networks}

\author{Yilin~Li,~\IEEEmembership{Student~Member,~IEEE,}
        Jian~Luo,
        Richard A. Stirling-Gallacher,~\IEEEmembership{Member,~IEEE,}
        and~Giuseppe~Caire,~\IEEEmembership{Fellow,~IEEE}
\thanks{Yilin Li is with the German Research Center, Huawei Technologies Duesseldorf GmbH, 80992 Munich, Germany, and with the Communications and Information Theory Group, Technische Universit{\"a}t Berlin, 10587 Berlin, Germany (e-mail: halodiplomat@gmail.com).}%
\thanks{Jian Luo and Richard A. Stirling-Gallacher are with the German Research Center, Huawei Technologies Duesseldorf GmbH, 80992 Munich, Germany (e-mail: jianluo@huawei.com; richard.sg@huawei.com).}%
\thanks{Giuseppe Caire is with the Communications and Information Theory Group, Technische Universit{\"a}t Berlin, 10587 Berlin, Germany, and with the Department of Electrical Engineering, The University of Southern California, Los Angeles, CA 90089, USA (e-mail: caire@tu-berlin.de).}%
}

\maketitle

\begin{abstract}
By allowing a large number of links to be simultaneously transmitted, directional antenna arrays with beamforming have been envisioned as a promising candidate to reach unprecedented levels of spatial isolation. To achieve the high efficiency of spatial reuse in improving system performance, an optimization problem that maximizes the achievable data rate of a multihop heterogeneous network, which incorporates the concept of integrated access and backhaul and supports both downlink and uplink transmissions on access and backhaul links, is formulated. The optimization problem is then systematically decomposed and demonstrated as NP-hard, and a heuristic joint scheduling and resource allocation algorithm is proposed to maximize the achievable data rate. In addition, an efficient dynamic routing algorithm is proposed to further enhance the data rate. With extensive system-level simulations, it is demonstrated that the proposed algorithms achieve significant gain over benchmark schemes, in terms of data rate, and closely approach the theoretical optimum, yet with lower latency. Besides, the proposed algorithms enable a flexible adjustment of downlink and uplink transmission duration allocation and support both half- and full-duplex modes with considerable performance enhancement. In particular, the proposed algorithms are capable of fulfilling different performance requirements for both point-to-point and point-to-multipoint communications.
\end{abstract}

%
%
\IEEEpeerreviewmaketitle

\section{Introduction}\label{Ch3_Sec1}
5G cellular communications are embracing mm-wave frequencies between $6$ and $\SI{100}{\giga\hertz}$, where the availability of large chunks of untapped bandwidth makes it possible to support stringent data rate requirements for future cellular systems~\cite{Rappaport}. Carrier frequencies up to $\SI{52.6}{\giga\hertz}$ with a bandwidth per single carrier up to $\SI{400}{\mega\hertz}$ has been standardized by 3GPP new radio (NR)~\cite{3GPPNROD}. Nevertheless, mm-wave signals suffer from increased isotropic pathloss and can be severely vulnerable to blockage, which results in outages and intermittent channel quality~\cite{YLiCommag}.

The combination of high propagation attenuation and blockage phenomenon advocates for a high-density deployment of infrastructure nodes~\cite{Bhushan}. In this regard, heterogeneous networks (HetNets), where a core macrocell seamlessly cooperates with small cells, have been treated as an available realization of network densification. Nevertheless, equipping all small cells with high performance fiber-based backhaul seems to be economically infeasible. As an attractive cost-efficient substitute to the wired backhaul, self-backhauling, which has been investigated by 3GPP NR as a part of the integrated access and backhaul (IAB) study item~\cite{3GPPNROD}, provides coverage extension and capacity expansion to fully exploit the heterogeneity of HetNets~\cite{Ge}. 

Another encouraging approach to cope with high isotropic pathloss and the sensitivity to blockage effect is the exploitation of beamforming techniques that form narrow beams with high antenna gain for data transmissions~\cite{RappaportTC}. This is possible as the small wavelength, which is one of the distinctive features of mm-wave bands, allows a large number of antenna arrays to be placed in a compact form factor. In general, directional antennas with beamforming reduce multi-user interference, where multiple links can be simultaneously transmitted to fully exploit spatial multiplexing gain. 

Hereof, how to maximize the performance of HetNets with self-backhauling and directional transmission becomes an interesting issue, particularly on the design of link scheduling, resource allocation, and path selection. A naive scheduling that lets the macrocell base station (BS) serve all users in a round robin fashion is neither practical nor efficient~\cite{Yuan}. By contrast, the limited interference at mm-wave bands makes it possible to schedule simultaneous transmissions, where the same radio resource can be allocated to multiple links to improve spatial reuse~\cite{YLiWCNC}. At the same time, when the backhaul link, which connects the associated small cell access point (AP) of a user to the macrocell BS, is weaker compared to the backhaul links to other nearby APs, a dynamic multihop routing scheme is much more favorable to improve overall performance.

\subsection{Related Works}\label{Ch3_Sec1_RW}
Increasing cellular capacity through self-backhauled small cells for IAB has become the primary motivation of many previous works. Some studies emphasized the placement of relay nodes inside cells to improve the signal quality at cell edges~\cite{BLi15, Tabassum}. More general approaches to enhance cellular capacity with multihop backhauling have been considered in~\cite{García-Rois, Sharma}. However, these works limit the number of links at each network node (BS, AP, user equipment (UE), etc.) to a single steerable beam, which does not take advantage of potential spatial reuse provided by highly directional mm-wave antenna arrays. 

When it comes to the exploitation of spatial reuse, the appropriate design of efficient scheduling policy has been suggested as a key challenge in realizing full benefits of multiplexing gain brought by simultaneous links~\cite{Qiao1, HWang, YLiuTWC, Niu2015, Niu2}. Nevertheless, with the exception of some studies that cover limited models~\cite{Niu2015, Niu2}, none of the aforementioned works to date has considered the degree of spatial link isolation. In other words, the interference in mm-wave communications has a much weaker effect than in sub-$\SI{6}{\giga\hertz}$ networks, but not negligible. We think that the question of whether the hypothesis of fully isolated pseudo-wired like links in certain scenarios of mm-wave communications is realistic remains open, as also addressed in~\cite{García-Rois}. The available capacity on simultaneously transmitted link should not be always considered as the same.

Furthermore, recent research efforts, particularly for mm-wave HetNets, have addressed versatile aspects of resource allocation, including frequency resource allocation~\cite{NWang}, power allocation~\cite{Hao}, joint power and time allocation~\cite{YLiu}, and joint scheduling and power allocation optimization~\cite{Goyal}. Besides the consideration of resource allocation, authors in~\cite{Yuan} proposed a polynomial time algorithm for joint scheduling and routing, unlike traditional NP-hard solutions, by extending the work in~\cite{Hajek}. Even though the aforementioned research studies cover a set of resource allocation and routing schemes that enhance system performance under different network configurations and constraint models, the authors focused on either one aspect of resource allocation (e.g., frequency allocation in~\cite{NWang} and power allocation in~\cite{Hao}), or limited combination (transmission duration and power allocation in~\cite{Goyal}, or scheduling and routing in~\cite{Yuan}). The joint solution of scheduling, resource allocation, and routing, have not been considered in any of above works. By contrast, we formulate a joint optimization model with an accurate schedule-dependent representation of data rate taking into account the resource allocation and routing to derive our main results.

\subsection{Contributions}\label{Ch3_Sec1_Ctrbt}
In this paper, we apply binary interference classification, i.e., an interference condition that prevents two links from being simultaneously active and is widely used in the literature (\cite{García-Rois, Niu2015, Niu2, Yuan, Hajek, AZhou}) for designing scheduling algorithms, to the joint scheduling, and resource allocation, and routing optimization. Nevertheless, on top of the binary classification, we formulate the data rate of scheduled links as a function that depends on actual signal-to-interference-plus-noise-ratio (SINR), where links that experience inter-link interference are not blocked completely. Specifically, multiple links are allowed to be simultaneously transmitted provided that mutual interferences among these links are below than a configurable threshold, which is usually considered in practical physical layer implementation. The main contributions of this paper are summarized as follows:
\begin{itemize}
	\item We formulate the optimization problem of joint scheduling and resource allocation (JSRA) for a typical HetNet with multihop IAB structure into a mixed integer nonlinear programing (MINLP) problem, in which the data rate, determined by scheduling, transmission duration, and power constraints, as well as by actual interference, is maximized by fully enabling simultaneous transmission to harvest spatial multiplexing gain.
	\item The constrained optimization problem is then demonstrated to be NP-hard. In order to obtain feasible solution, we systematically decompose the problem into three sub-problems, namely simultaneous transmission scheduling, transmission duration allocation, and transmission power allocation. Based on this, we propose a heuristic scheduling algorithm referred to as conflict graph maximum independent set (CG-MIS) algorithm, a proportional fair transmission duration algorithm, and water-filling power allocation algorithm to solve the optimization problem based on fixed routing decision. 
	\item The fixed routing decision, however, may degrade the system performance when the traffic from different users is congested at some network nodes. Therefore, we propose a dynamic routing (DR) algorithm, where the selection of path between BS and UE depends on real-time network statistics and the user traffic is routed along lightly loaded links, to investigate the ability of further improvement in the data rate achieved by the above scheduling and resource allocation algorithms.
	\item Extensive simulations have been conducted under numerous system parameters to demonstrate that the proposed algorithms outperform benchmark schemes for multiplexing and interference mitigation, in terms of achievable data rate, and closely approach the theoretical optimum, yet with lower latency. The system performance of the proposed algorithms under different frame structures and duplex schemes are also analyzed. In particular, the potential of the proposed algorithms in boosting the system performance of point-to-point (P2P) communications, has also been studied as an extension to the point-to-multipoint (P2M) scenario for the considered HetNets.
\end{itemize}

The remainder of this paper is organized as follows: Section~\ref{Ch3_Sec2} presents the system model and formulates the joint optimization problem of scheduling and resource allocation. The link scheduling algorithm, the transmission duration allocation algorithm, and the transmission power allocation algorithm are proposed in Section~\ref{Ch3_Sec3}. In Section~\ref{Ch3_Sec4}, we introduce the routing algorithm. The performance of the proposed algorithms are evaluated by extensive simulations in Section~\ref{Ch3_Sec5}, followed by a summary concluding this paper in Section~\ref{Ch3_Sec6}.

A conference version of this paper has appeared in~\cite{YLiWCNC}. The current paper extends the previous work with the design of routing algorithm and the application of the proposed joint scheduling, resource allocation, and routing algorithms to P2P communications with focus on data delivery in vehicle platoon. The current paper also includes all the derivations, discussion of the extensions, and more detailed simulations.

\section{System Overview and Problem Formulation}\label{Ch3_Sec2}
In this section, we first introduce the network model, the available connection between BS and UE, and the frame structure that incorporates the concept of space-division multiple access (SDMA) group, in which multiple links are simultaneously scheduled, in Section~\ref{Ch3_Sec2_SubSec1}. Then, we show the channel model that is applied to the calculation of achievable link capacity in Section~\ref{Ch3_Sec2_SubSec2}. Finally, the JSRA optimization problem is formulated in Section~\ref{Ch3_Sec2_SubSec3}, where the complexity of the optimization problem is demonstrated to be NP-hard, such that we are motivated to develop heuristic algorithms to solve the problem efficiently. The important notations and system parameters defined in this section are summarized in Table~\ref{Ch3_Table_SysMod} and will be used in the rest of this paper.
\renewcommand{\arraystretch}{0.6}
\begin{table}[tbp]
	\centering
	\caption{System Model Parameters}
	\label{Ch3_Table_SysMod}
		\begin{tabular}{|l|l|l|}
			\hline
			Notation & Description & Value \\ \hline
			$T$ & Frame length & $\SI{10}{\milli\second}$ \\ \hline
			$l_i$ & Pathloss of a link at carrier frequency $f$ and distance $d_i$ & See~\eqref{eqnPL2}  \\ \hline
			$f$ & Carrier frequency & $\SI{28}{\giga \hertz}$ \\ \hline
			$c$ & Speed of light & $\SI{3e8}{\meter/\second}$ \\ \hline
			$n_\text{L}$ & Pathloss exponent & LOS, NLOS: $2.1,3.17$ \\ \hline
			SF & Shadowing factor & LOS, NLOS: $2.38,6.44$ \\ \hline
			$p{(d_i)}$ & Probability of a link with distance $d_i$ to be LOS & See~\eqref{eqnLOSNLOS2}  \\ \hline
			$d_1$ & Parameters in $d_{1}/d_{2}$ model & $20$ \\ \hline
			$d_2$ & Parameters in $d_{1}/d_{2}$ model & $39$ \\ \hline
			$g_i$ & Antenna gain of link $i$ (Antenna array vertical $\times$ horizontal) & BS/AP, UE:  $16\times 8$, $4\times 4$ \\ \hline
			$\text{SINR}_i$ & SINR of link $i$ & See~\eqref{eqnSINR} \\ \hline
			$p_i$ & Transmission power of link $i$ & See~\eqref{eqnoptPS} \\ \hline
			$\eta$ & Thermal noise power & $\SI{2e-11}{\watt}$ \\ \hline
			$\delta_{i}^{(k)}$ & Schedule indicator of link $i$ in SDMA group $k$ & See~\eqref{eqnSI}\\ \hline
			$r_i$ & Channel capacity of link $i$ & See~\eqref{eqnCC} \\ \hline
			$b_i$ & Allocated bandwidth of link $i$ & See~\eqref{eqnbw} \\ \hline
			$n^{(k)}$ & Number of slots in SDMA group $k$ & See~\eqref{eqnAllSlot} \\ \hline
			$\mathcal{{M}}_{s,k}$ & Set of simultaneously transmitted links from node $s$ in SDMA group $k$  & See~\eqref{eqnpower} \\ \hline
			$P_\text{max}$ &  Maximum transmission power of BS/AP/UE & $\SI{1}{\watt}$ for BS/AP, UE: $\SI{1}{\watt}$, $\SI{0.1}{\watt}$ \\ \hline
			$\sigma$ & Inter-link interference threshold & $\SI{1e-8}{\watt}$ \\ \hline
			$B$ & System bandwidth & $\SI{1}{\giga \hertz}$ \\ \hline
	\end{tabular}
\end{table}

\subsection{Network, Connection, and Frame Structure}\label{Ch3_Sec2_SubSec1}
We consider a typical HetNet that consists of a macrocell BS, a set of small cell APs, and UEs that are associated either directly with the BS, or with the geographically closest AP and connected to the BS via multihop. 
Fig.~\ref{Ch3_Fig_Graph} shows an example of the considered HetNet with one BS, two APs and four UEs. We represent the network as a directed graph $\mathcal{G}(\mathcal{V},\mathcal{E})$, where $\mathcal{V}$ indicates the set of nodes (BS, APs and UEs) and $\mathcal{E}$ indicates the set of links. For the exampled network in Fig.~\ref{Ch3_Fig_Graph}, an abstracted graph model, which we refer to as \textit{link graph}, is also illustrated in Fig.~\ref{Ch3_Fig_Graph}. Admissible connections are BS$\rightleftharpoons$AP, BS$\rightleftharpoons$UE, and AP$\rightleftharpoons$UE, with both downlink and uplink traffic flows along arbitrary routes. 
\begin{figure}[tbp]
	\centering
	\includegraphics[width=0.9\linewidth]{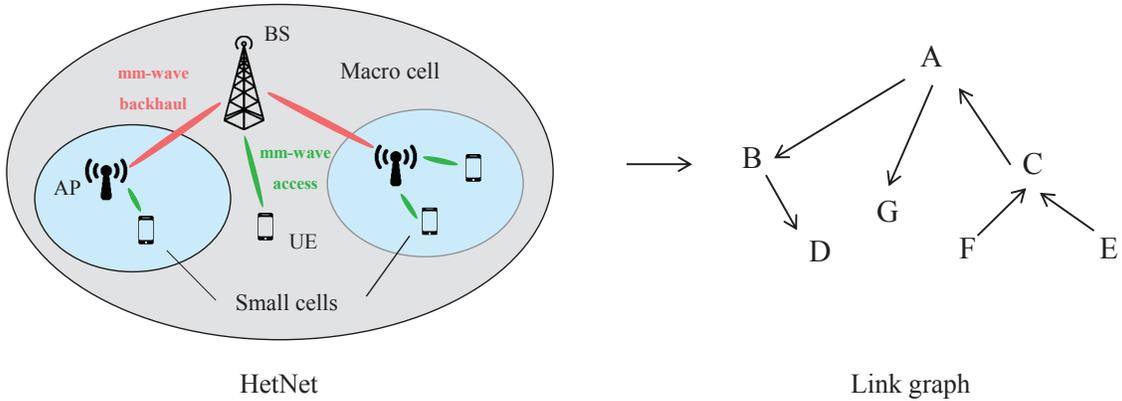}
	\caption{An exampled HetNet consisting of one BS, two APs, and four UEs is presented by a directed graph.} \label{Ch3_Fig_Graph}
\end{figure}

As in~\cite{YLiCommag}, BS and AP support mm-wave bands for backhaul and access transmission and reception (in-band backhauling). The transmission requests and corresponding time/frequency synchronization information is assumed to be collected by sub-$\SI{6}{\giga\hertz}$ communications. For the data transmission of each UE associated with APs, a predefined route is selected, where we are able to focus on the JSRA optimization. Nevertheless, the design of DR is introduced in Section~\ref{Ch3_Sec4} as an extension to further improve the network performance. 

We further consider a time-division duplex mm-wave frame structure as shown in Fig.~\ref{Ch3_Fig_Frame}, where the system time is divided into consecutive frames with period $T$. Each cycle begins with a beacon phase, followed by a data transmission phase modeled as a slotted-based time period, in which transmissions between any valid pair of nodes can be scheduled.
\begin{figure}[tbp]
	\centering
	\includegraphics[scale=0.8]{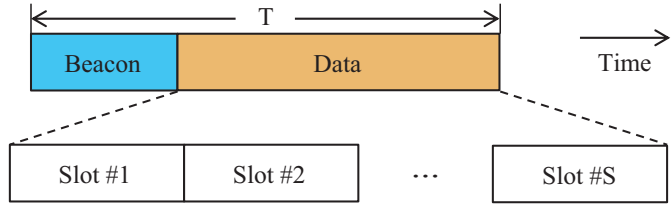}
	\caption{Frame structure including a slotted data transmission phase.} \label{Ch3_Fig_Frame}
\end{figure}

Time-division multiple access (TDMA) is widely adopted for mm-wave channel access in 5G networks~\cite{Qiao1, Niu2, Niu2015}, where within the period of each frame, it is assumed that network topology and channel condition remain unchanged. 
In TDMA scheme, each slot is exclusively occupied by a single link. By enabling the possibility of spatial multiplexing, multiple simultaneous transmissions can be scheduled in each slot. Hence, we can allocate more transmission duration to each link, such that the achievable data rate of each link is improved without other sophisticated techniques. An example of the comparison of TDMA and simultaneous transmission in slot allocation for six transmission links in a frame with a ten-slot data phase is illustrated in Fig.~\ref{Ch3_Fig_TDMA_simul}, where the corresponding number of allocated slots are written below the links.
\begin{figure}[tbp]
	\centering
	\includegraphics[scale=0.85]{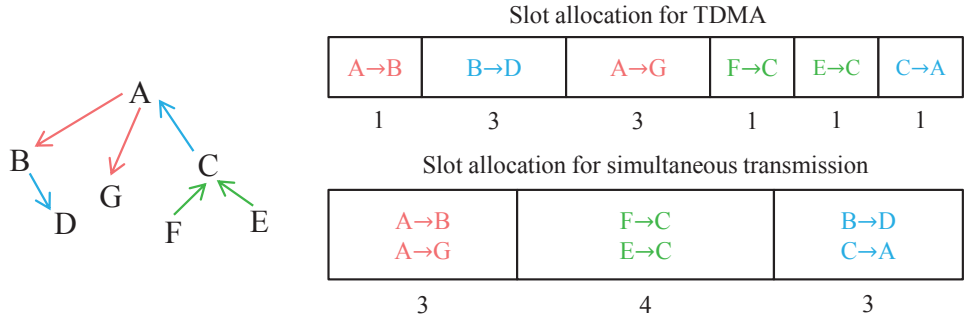}
	\caption{An example of slot allocation for TDMA and simultaneous transmission.} \label{Ch3_Fig_TDMA_simul}
\end{figure}

\subsection{Channel, Traffic, and Link SINR}\label{Ch3_Sec2_SubSec2}
We define a \textit{SDMA group} as a transmission interval that consists of consecutive slots allocated to a link when simultaneous transmission is enabled. For simplicity, in the rest of the paper we use the term ``group'' to represent the SDMA group. In the example depicted in Fig.~\ref{Ch3_Fig_TDMA_simul}, the data phase consisting of ten slots is separated in three groups, in which different links are simultaneously scheduled. However, each link can be scheduled only once in a frame. 

For an accurate generation of link capacity, we compute the isotropic pathloss for link $i$ transmitted from node $m$ towards node $n$ at distance $d_i$ in meters, denoted as $l_i$, given by~\cite{RappaportTC} as
\begin{equation} 
l_i = \left(\frac{4{\pi}f}{c}\right)^2 \cdot {d_i}^{n_\text{L}} \cdot \text{SF}, \label{eqnPL2}
\end{equation}
where link $i$ is determined to be LOS or NLOS with probability $p{(d_i)}$ according to~\cite{RappaportTC, YLiTWC} as
\begin{equation} 
p{(d_i)} = \min\left(\frac{d_1}{d_i},1\right)\left(1-e^{-{d_i}/{d_2}}\right) + e^{-{d_i}/{d_2}}. \label{eqnLOSNLOS2}
\end{equation}
Here, $f$ indicates the carrier frequency in $\si{\hertz}$, and $n_\text{L}$ represents the pathloss exponent. $c$ is the speed of light. The impact of objects such as trees, cars, etc.\ is modeled separately using the shadowing factor (SF). Further, an additive white Gaussian noise (AWGN) channel is assumed within all links, and the channel knowledge is assumed to be available at the BS. 

We further apply similar modeling of antenna gain as in~\cite{YLiTWC} and denote the antenna gain of link $i$ as $g_i$. Unlike the recent works~\cite{Qiao1, HWang, Niu2}, which have assumed a pseudo-wired behavior for mm-wave links, we do not assume that the interference is negligible but rather compute real mutual interference between simultaneous links within each frame. Then, the instantaneous SINR of link $i$, denoted as $\text{SINR}_{i}$, is provided by
\begin{equation}
\text{SINR}_{i} = \frac{p_i g_i l_i^{-1}}{\eta + \sum_{j}p_j g_j l_j^{-1}},
\label{eqnSINR}
\end{equation}
where $p_i$, $g_i$, $l_i$, and $\eta$ represent the transmission power, the antenna gain, the pathloss, and the thermal noise power of link $i$, respectively. The experienced interference of link $i$ is model by $\sum_{j}p_j g_j l_j^{-1}$ which summarizes the receive power of all interfered link $j$ at link $i$. 

\subsection{Problem Formulation}\label{Ch3_Sec2_SubSec3}
We consider $M$ transmission links to be scheduled in a given frame that consists of $N$ slots. 
For each link $i$, we defined a logic indicator $\delta_{i}^{(k)}$ that controls the schedule policy of whether link $i$ is scheduled in group $k$. Specifically, the policy is modeled as
\begin{equation} 
\delta_{i}^{(k)} = \begin{cases}
1, & \text{Link $i$ is scheduled in group $k$,} \\
0, & \text{Otherwise.}
\end{cases}
\label{eqnSI} 
\end{equation}

Then, the actual capacity of link $i$, denoted as $r_{i}$, can be represented according to Shannon channel capacity equation as
\begin{equation}
r_{i} = b_i \cdot \log{ \left( 1 + \delta_{i}^{(k)} \text{SINR}_{i} \right) },
\label{eqnCC}
\end{equation}
where $b_i$ indicates the available bandwidth at link $i$ in $\si{\hertz}$. We further assume that the $N$ slots of the given frame are allocated into $K$ groups, and the number of slots allocated in each group is denoted as $n^{(k)}$.

%
%
Define a scheduling policy $\bm{\delta}$, a slot allocation policy $\bm{n}$, and a power allocation policy $\bm{p}$ as the set of binary vector $\delta_{i}^{(k)}$, the set of slot allocation ${n}^{(k)}$, and the set of power allocation ${p}_i$, $\forall k \in \{1,\dotsc,K\}$ and $\forall i \in \{1,\dotsc,M\}$, respectively. Then, we can define the achievable data rate for all the $M$ links in the considered frame with $K$ groups, given the scheduling policy $\bm{\delta}$, the slot allocation policy $\bm{n}$, and the power allocation policy $\bm{p}$, as
\begin{equation}
\sum_{i=1}^{M} \sum_{k=1}^{K} \frac{r_{i}n^{(k)}}{N}.
\label{eqnTP}
\end{equation}
As the total number of slots $N$ in each frame is fixed, then the objective function, which is the achievable data rate defined in~\eqref{eqnTP} and to be maximized, can be further simplified as
\begin{equation}
\sum_{i=1}^{M} \sum_{k=1}^{K} r_{i}n^{(k)}.
\label{eqnobj}
\end{equation}

Now, we analyze the system constraints of the optimization problem. First of all, the scheduling policy $\bm{\delta}$ is ruled by the following two constraints:
\begin{equation}
\sum_{k=1}^{K} \delta_{i}^{(k)} = 1,\; \forall i \in \{1,\dotsc,M\},
\label{eqnschd}
\end{equation}
and
\begin{equation}
\delta_{i}^{(k)} + \delta_{j}^{(k)} \leq 1,\; \forall\; \text{sequential link $i$ and $j$},\; \forall i,j \in \{1,\dotsc,M\}, \forall k \in \{1,\dotsc,K\}.
\label{eqnseq}
\end{equation}
Here,~\eqref{eqnschd} indicates that each link can be scheduled only once in each frame (in one of the $K$ groups), as demonstrated in Section~\ref{Ch3_Sec2_SubSec2}. However, each link scheduled in group $k$ can be allocated with multiple slots, namely $n^{(k)}$. Further, due to half-duplex constraint, the sequential links (e.g.\ backhaul and access links as edge $A \rightarrow B$ and $B \rightarrow D$ in Fig.~\ref{Ch3_Fig_Graph}) cannot be scheduled in the same group, which is governed by the constraint demonstrated in~\eqref{eqnseq}.

Next, the constraint on the slot allocation policy $\bm{n}$, where the total number of allocated slots in all groups should not be larger than $N$, is represented as
\begin{equation}
\sum_{k=1}^{K} n^{(k)} \leq N,\; \forall k \in \{1,\dotsc,K\}.
\label{eqnslot}
\end{equation}

Lastly, the summarized allocated power of the links transmitted from the same node, say node $s$, in group $k$, should not exceed the total available transmission power of the transmitter. Denoting the set of simultaneously transmitted links from node $s$ in group $k$ as $\mathcal{{M}}_{s,k}$, the constraint on the power allocation policy $\bm{p}$ is given by
\begin{equation}
\sum_{i \in \mathcal{M}_{s,k}} p_{i} \leq P_\text{max},
\label{eqnpower}
\end{equation}
where $P_\text{max}$ indicates the maximum transmission power of BS/AP/UE.

Finally, we can formulate our JSRA optimization problem to maximize the data rate $r_{i}$, under the constraints of scheduling policy $\bm{\delta}$, slot allocation policy $\bm{n}$, and power allocation policy $\bm{p}$, as follows:
\begin{problem}
	(JSRA optimization)
	\begin{align*}
	\underset{\bm{\delta},\bm{n},\bm{p}}\max &\null \quad \sum_{i=1}^{M} \sum_{k=1}^{K} r_{i} n^{(k)}, \\
	\textnormal{s.t.} &\null \quad \textnormal{ constraints}~\eqref{eqnschd}\text{--}\eqref{eqnpower}. \numberthis \label{eqnopt}
	\end{align*}\label{Ch3_prb1}
\end{problem}
\vspace{-2.5\topsep}
The maximization problem indicated in~\eqref{eqnopt} with constraints~\eqref{eqnschd}--\eqref{eqnpower} is a MINLP problem~\cite{JLee}, where the variables are categorized as
\begin{itemize}
	\item Integer variables: $n^{(k)}$,
	\item Continuous variables: $p_{i}$,
	\item Binary variables: $\delta_{i}^{(k)}$.
\end{itemize}
In general, one among the simplest MINLP problem is the 0--1 Knapsack problem, which is a class of problems and proven to be typically NP-hard~\cite{JLee}. Nevertheless, the multiplication of the above variables in~\eqref{eqnopt} further complicates the proposed optimization problem~\ref{Ch3_prb1} and makes it even more complex than the 0--1 Knapsack problem. Specifically, there exists the third-order term $\delta_{i}^{(k)}p_{i}n^{(k)}$ in~\eqref{eqnopt}, in which the coupling exists among the constraints where slot and power allocation policies rely on the results of link scheduling. Therefore, the considered optimization problem~\ref{Ch3_prb1} is NP-hard.
In the next section, we propose a heuristic JSRA algorithm to decouple the correlated variables and solve problem~\ref{Ch3_prb1} efficiently with low complexity.

\section{Scheduling and Resource Allocation Algorithm}\label{Ch3_Sec3}
As mentioned in Section~\ref{Ch3_Sec2}, the scheduling policy $\bm{\delta}$, the slot allocation policy $\bm{n}$, and the power allocation policy $\bm{p}$ are three key and correlated terms for solving the optimization problem~\ref{Ch3_prb1}. However, the slot and power allocation policy cannot be determined until the scheduling decision is made. Specifically, the slot allocation depends on the demand of links that are scheduled in each SDMA group, and the number of scheduled links affects the power allocation policy of the corresponding group. Therefore, in this section we propose a heuristic JSRA algorithm to solve problem~\ref{Ch3_prb1}, where simultaneous transmission is fully exploited to increase data rate. 

\subsection{CG-MIS Scheduling Algorithm}\label{Ch3_Sec3_SubSec1}
To better demonstrate the relationship among different links presented by constraints~\eqref{eqnschd} and~\eqref{eqnseq}, we introduce an appropriate concept named \textit{conflict graph} (CG). Different from the link graph depicted in Fig.~\ref{Ch3_Fig_Graph}, in a conflict graph, each node represents one link in the network, and there is an edge connecting two nodes if ``conflict'' exists. Specifically, the conflicts can be derived from sequential links (e.g., backhaul and access links) as described in~\eqref{eqnseq}, due to half-duplex constraint, or from links that interfere with each other severely. For the latter case, we define an interference threshold $\sigma$ as the criterion that judges whether an edge exists between two non-sequential links (nodes)\footnote{The interference information is assumed to be acquired from e.g.\ interference sensing procedure. This issue, as well as other signaling designs, have been addressed by the other works of us and are beyond the scope of this paper.}: An edge exists if the SINR of any of the two corresponding links is larger than $\sigma$. An example of CG construction is illustrated in Fig.~\ref{Ch3_Fig_ConGraph}.
\begin{figure}[tbp]
	\centering
	\includegraphics[scale=0.9]{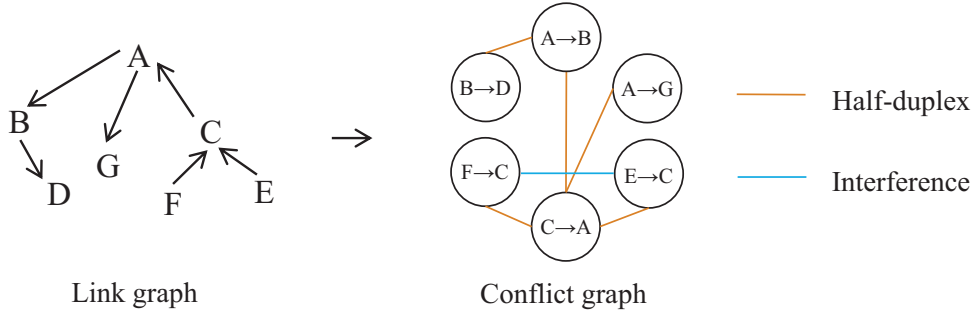}
	\caption{Conflict graph construction.} \label{Ch3_Fig_ConGraph}
\end{figure}

Based on the CG, we propose a maximum independent set (MIS) based scheduling algorithm, namely the CG-MIS algorithm, to distribute links into different groups. In graph theory, an independent set is a subset of nodes in a graph, in which no pair of nodes are adjacent. A maximal independent set is either an independent set such that adding any other node to the set forces the set to contain an edge or the set of all nodes of an empty graph. It is demonstrated in~\cite{Karp} that the computational complexity of finding the MIS of a general graph is NP-hard, and there are no very efficient algorithms that are capable to find optimal solution in polynomial time. Therefore, we utilize the minimum-degree greedy algorithm to solve the problem of finding the MIS of a CG~\cite{Halldorsson}. 

We denote the CG as $\mathcal{G}_\textnormal{C}(\mathcal{V}_\textnormal{C},\mathcal{E}_\textnormal{C})$ , where $\mathcal{V}_\textnormal{C}$ and $\mathcal{E}_\textnormal{C}$ represent the set of nodes and the set of edges in the CG, respectively. For any node $v \in \mathcal{V}_\textnormal{C}$, we define its neighbors as $\mathcal{N}_{v}$, which consists of all adjacent nodes of $v$. The degree of any node $v \in \mathcal{V}_\textnormal{C}$ is denoted by $\Delta_{v}$. Then, the CG-MIS scheduling algorithm is summarized in Algorithm~\ref{Ch3_Alg1}. In the algorithm, links are iteratively scheduled into each group until all links have been traversed, as indicated in line~\ref{Ch3_Alg1:1}. At the beginning of each iteration, the set for recording the scheduled links in the corresponding group (say group $k$), denoted as $\mathcal{V}^{k}$, is initialized as empty. Line~\ref{Ch3_Alg1:5}--\ref{Ch3_Alg1:9} describe the minimum-degree greedy scheduling algorithm of group $k$. In line~\ref{Ch3_Alg1:10}, the traversed links are removed from the set $\mathcal{V}_\textnormal{C}$, which is later used for updating the non-traversed set $\mathcal{V}^{k,u}$ for next group in line~\ref{Ch3_Alg1:4}.
\begin{algorithm}[tbp]
	\SetAlgoLined
	\renewcommand\baselinestretch{1}\selectfont
	\KwIn{Conflict graph $\mathcal{G}_\textnormal{C}(\mathcal{V}_\textnormal{C},\mathcal{E}_\textnormal{C})$}
	\KwOut{Scheduling policy $\bm{\delta}$}
	\begin{itemize}
		{\item $\mathcal{V}_\textnormal{C}$: Set of nodes}
		{\item $\mathcal{E}_\textnormal{C}$: Set of edges}
		{\item $\mathcal{V}^{k}$: Set of scheduled links in group $k$}
		{\item $\mathcal{V}^{k,u}$: Set of unscheduled links in group $k$ }
		{\item $\Delta_{v}$: Degree of node $v$}
		{\item $\mathcal{N}_{v}$: Neighbors of node $v$}
		{\item $k$: Iterator (index of group) }
	\end{itemize}
	\textit{Initialization}:
	$k=0$ \\
	\Begin{
		\everypar={\nl}
		\While{$\mathcal{V}_\textnormal{C}\neq\varnothing$}{\label{Ch3_Alg1:1}
			$k=k+1$\;
			$\mathcal{V}^{k}=\varnothing$\;
			$\mathcal{V}^{k,u}=\mathcal{V}_\textnormal{C}$\; \label{Ch3_Alg1:4}
			\While{$\mathcal{V}^{k,u}\neq\varnothing$}{\label{Ch3_Alg1:5}
				Get $v \in \mathcal{V}^{k,u}$ where $\Delta_{v}=\min_{v' \in \mathcal{V}^{k,u}} \Delta_{v'}$\; \label{Ch3_Alg1:6}
				$\mathcal{V}^{k} = \mathcal{V}^{k} \cup v$\; \label{Ch3_Alg1:7}
				$\mathcal{V}^{k,u} = \mathcal{V}^{k,u} - {\{v \cup \mathcal{N}_{v}\}}$\; \label{Ch3_Alg1:8}
			} \label{Ch3_Alg1:9}
			$\mathcal{V}_\textnormal{C} = \mathcal{V}_\textnormal{C} - \mathcal{V}^{k}$\;\label{Ch3_Alg1:10}
		}
		\textbf{Return} $\mathcal{V}^{k}$ for each group $k$\;
	}
	\caption{CG-MIS Link Scheduling Algorithm}\label{Ch3_Alg1}
\end{algorithm}

The computational complexity of the algorithm is $\mathcal{O}(|{\mathcal{V}_\textnormal{C}}|^2)$, compared to $\mathcal{O}(|{\mathcal{V}_\textnormal{C}}|^2 \cdot 2^{|{\mathcal{V}_\textnormal{C}}|^2})$ of naive brute force scheme, where $|{\mathcal{V}_\textnormal{C}}|$ indicates the cardinality of the node set $\mathcal{V}_\textnormal{C}$, namely the number of nodes in the set. The performance analysis of the algorithm is presented as below.
\begin{lemma}
	Minimum-degree greedy algorithm outputs an independent set $\{\mathcal{V}^{k}|k=1,\dotsc,K\}$ such that $|\mathcal{V}^{k}| \geq \frac{|\mathcal{V}_\textnormal{C}|}{\Delta + 1}$ where $\Delta$ is the maximum degree of any node in the graph. \label{Ch3_Lem1}
\end{lemma}
\begin{proof}
	Denote $\mathcal{V}_\textnormal{C} \backslash \mathcal{V}^k$ as the complement of set $\mathcal{V}^k$ for set $\mathcal{V}_\textnormal{C}$, namely $\mathcal{V}_\textnormal{C} = \mathcal{V}^k \cup \mathcal{V}_\textnormal{C} \backslash \mathcal{V}^k$. Then, we upper bound the number of nodes in $\mathcal{V}_\textnormal{C} \backslash \mathcal{V}^k$, i.e. $|\mathcal{V}_\textnormal{C} \backslash \mathcal{V}^k|$, as follows: A node $u$ is in $\mathcal{V}_\textnormal{C} \backslash \mathcal{V}^k$ because it is removed as a neighbor of some node $v \in V^k$ when greedily added $v$ to $\mathcal{V}^k$. Associate $u$ to $v$. A node $v \in \mathcal{V}^k$ can be associated at most $\Delta$ times as it has at most $\Delta$ neighbors. Hence, we have $|\mathcal{V}_\textnormal{C} \backslash \mathcal{V}^k| \leq \Delta |\mathcal{V}^k|$. Then, as every node is either in $\mathcal{V}^k$ or $\mathcal{V}_\textnormal{C} \backslash \mathcal{V}^k$, we have $|\mathcal{V}_\textnormal{C} \backslash \mathcal{V}^k| + |\mathcal{V}^k| = |\mathcal{V}_\textnormal{C}|$ and therefore $(\Delta+1)|\mathcal{V}^k| \geq |\mathcal{V}_\textnormal{C}|$, which implies that $|\mathcal{V}^k| \geq \frac{|\mathcal{V}_\textnormal{C}|}{\Delta+1}$.
\end{proof}
\begin{corollary}
	Minimum-degree greedy algorithm gives a $\frac{1}{\Delta+1}$ approximation for MIS in graphs of degree at most $\Delta$. \label{Ch3_Cor1}
\end{corollary}
\begin{proof}
	A straightforward result of Lemma~\ref{Ch3_Lem1}.
\end{proof}
As the maximum degree of a graph varies with the graph topology, we are interested in the performance comparison between greedy algorithm and brute force scheme in general case, where the performance ratio can be characterized by the average degree of the graph, denoted as $\bar{d}$. This is presented in the following theorem:
\begin{theorem}
	Minimum-degree greedy algorithm achieves a $\frac{2\bar{d}+3}{5}$ performance ratio for MIS in graphs of average degree $\bar{d}$. \label{Ch3_Thm2}
\end{theorem}
\begin{proof}
	The proof is provided in~\cite{Halldorsson} in combination with a fractional relaxation technique and not addressed here due to space limitation.
\end{proof}

\subsection{Slot Allocation Algorithm}\label{Ch3_Sec3_SubSec2}
With the simultaneous transmission scheduling policy, we propose a proportional fair time resource allocation scheme to determine the transmission duration (slot) for each group. This algorithm gets the required slot of each link by TDMA scheme (exampled in Fig.~\ref{Ch3_Fig_TDMA_simul}) in a group, calculates the maximal number of all the required slots, and proportionally allocate all slots in a frame to the links in each group. As more time slots are allocated to each link, the achievable data rate of these links are increased.

We denote the number of slots required by link $i$ as $n_{i}$, then the maximal number of required slots among all links in group $k$, denoted as $n^{(k)}_\textnormal{max}$, can be obtained by
\begin{equation}
n^{(k)}_\textnormal{max}=\max\limits_{i \in \mathcal{V}^{k}}{n_{i}}. \label{eqnMaxSlot}
\end{equation}
Based on this, the total $N$ slots in a frame are allocated to each group proportionally to its maximum number of required slots $n^{(k)}_\textnormal{max}$. Hence, the number of slots allocated to all links in group $k$, denoted as $n^{(k)}$, can be calculated as
\begin{equation} 
n^{(k)}={\Biggl\lfloor}\frac{ n^{(k)}_\textnormal{max} }{ \sum_{k} n^{(k)}_\textnormal{max} }{\cdot}N{\Biggr\rfloor},
\label{eqnAllSlot} 
\end{equation}
where $\lfloor . \rfloor$ represents the floor function. Pseudo code of the time resource allocation algorithm is summarized in Algorithm~\ref{Ch3_Alg2}. In the algorithm, firstly, the required number of slots of all links scheduled in each group, $n_{i}$, are collected, as indicated in line~\ref{Ch3_Alg2:1}--\ref{Ch3_Alg2:3}. Then, for each group $k$, the maximal number of required slots, $n^{(k)}_\textnormal{max}$, are calculated as in line~\ref{Ch3_Alg2:4}. Based on this, the actual allocatable number of slots are determined for all groups, as indicated in line~\ref{Ch3_Alg2:5}--\ref{Ch3_Alg2:9}. As each link can only be scheduled into one group, the computational complexity of the algorithm is $\mathcal{O}(|\mathcal{V}_\textnormal{C}|)$.
\begin{algorithm}[tbp]
	\renewcommand\baselinestretch{1}\selectfont
	\SetAlgoLined
	\KwIn{Set of scheduled links in group $\mathcal{V}^{k}$, number of required slot $n_{i}$}
	\KwOut{Slot allocation policy $\bm{n}$}
	\begin{itemize}
		{\item $n^{(k)}_\textnormal{max}$: Maximum number of slots of links in group $k$}
		{\item $n^{(k)}$: Allocated number of slots for links in group $k$}
		{\item $k$: Iterator (index of group) }
	\end{itemize}
	\textit{Initialization}:
	$k=0$ \\
	\Begin{
		\everypar={\nl}
		\While{$k \neq K$}{
			$k=k+1$\;
			\ForEach{\textnormal{link} $i \in \mathcal{V}^{k}$ }{ \label{Ch3_Alg2:1}
				Get the number of required slots $n_{i}$\; \label{Ch3_Alg2:2}
			} \label{Ch3_Alg2:3}
			$n^{(k)}_\textnormal{max}=\max\limits_{i \in \mathcal{V}^{k}}{n_{i}}$\; \label{Ch3_Alg2:4}
		}
		$k=0$\;
		\While{$k \neq K$}{ \label{Ch3_Alg2:5}
			$k=k+1$\; \label{Ch3_Alg2:6}
			$n^{(k)}={\Bigl\lfloor}\frac{ n^{(k)}_\textnormal{max} }{ \sum\limits_{k} \label{Ch3_Alg2:7} 
				n^{(k)}_\textnormal{max} }{\cdot}N{\Bigr\rfloor}$\; \label{Ch3_Alg2:8}
		} \label{Ch3_Alg2:9}
		\textbf{Return} $n^{(k)}$ for each $k$\;
	}
	\caption{Slot Allocation Algorithm}\label{Ch3_Alg2}
\end{algorithm}

\subsection{Power Allocation Algorithm}\label{Ch3_Sec3_SubSec3}
By enabling spatial multiplexing, at some nodes (BS, AP), there will be multiple links to be simultaneously transmitted. Hence, power control is required at these nodes to fulfill the power splitting in P2M transmission situation. Specifically, after link scheduling and slot allocation, constraints~\eqref{eqnschd}--\eqref{eqnslot} are satisfied and consequently the optimization problem~\ref{Ch3_prb1} in~\eqref{eqnopt} can be reformed as
\begin{problem}
	(Power allocation optimization)
	\begin{align*}
	\underset{\bm{p}}\max &\null \quad \sum_{i \in \mathcal{V}^{k}} \sum_{k=1}^{K} r_{i} , \\
	\textnormal{s.t.} &\null \quad \sum_{i \in \mathcal{M}_{s,k}} p_{i} \leq P_\textnormal{max}. \numberthis \label{eqnoptP}
	\end{align*}\label{Ch3_prb2}
\end{problem}
Here, in~\eqref{eqnoptP}, the objective function first summarizes the achievable channel capacity, instead of data rates, as the slot allocation policy has been determined, of all links scheduled in group $\mathcal{V}^{k}$ (i.e., $\sum_{i \in \mathcal{V}^{k}}$) and further summarizes all groups of a frame (i.e., $\sum_{k=1}^{K}$). As the power allocation of links in one group is independent of that in the others, and the total transmission power of BS/AP remains unchanged for each group, the above maximization problem~\ref{Ch3_prb2} provided by~\eqref{eqnoptP} can be further relaxed as 
\begin{problem}
	(Relaxed power allocation optimization)
	\begin{align*}
	\underset{\bm{p}}\max &\null \quad \sum_{i \in \bm{M}_{s,k}} \log \left(1+\frac{g_i l_i^{-1}}{\eta} p_{i}\right) , \\
	\textnormal{s.t.} &\null \quad \sum_{i \in \mathcal{M}_{s,k}} p_{i} \leq P_\textnormal{max}. \numberthis \label{eqnoptPrlx}
	\end{align*}\label{Ch3_prb3}
\end{problem}
Here, we assume that the total system bandwidth $B$ of node $s$ is equally averaged to all links in $\mathcal{M}_{s,k}$, i.e.,
\begin{equation}
b_i = \frac{B}{|\mathcal{M}_{s,k}|}. \label{eqnbw}
\end{equation}
Hence, the term $b_i$ is eliminated in~\eqref{eqnoptPrlx} from~\eqref{eqnSINR}. Moreover, the term $\sum_{j}p_j g_j l_j^{-1}$ that describes the experienced interference power of link $i$ is also eliminated in~\eqref{eqnoptPrlx}, as all the other links scheduled in group $V^k$ are isolated from link $i$, namely these links lead to interference power that is less than the threshold $\sigma$, and are allowed to be simultaneously transmitted with link $i$ under tolerable interference level. Nevertheless, we still calculate the actual SINR of link $i$ for performance evaluation in Section~\ref{Ch3_Sec5} as addressed in Section~\ref{Ch3_Sec1_Ctrbt}.

Denoting $\frac{g_i l_i^{-1}}{\eta}$ as $\gamma_{i}$, which we refer to as channel quality, the relaxed optimization problem~\ref{Ch3_prb3} provided by~\eqref{eqnoptPrlx} can be solved by the following theorem (known as water-filling solution).
\begin{theorem}
	The optimal solution of problem~\ref{Ch3_prb3} is given by
	\begin{equation}
	p_{i}^{*} = \max\left\{ \frac{1}{\phi^{*}} - \frac{1}{\gamma_{i}}, 0\right\},
	\label{eqnoptPS}
	\end{equation}
	where $p_{i}^{*}$ is the optimal transmission power allocated to link $i$, and the optimal Lagrangian multiplier $\phi^{*}$ is given by
	\begin{equation}
	\sum\limits_{i \in \mathcal{M}_{s,k}}{\max\left\{ \frac{1}{\phi^{*}} - \frac{1}{\gamma_{i}}, 0\right\}} = P_\textnormal{max}.
	\label{eqnLGRG}
	\end{equation}
	\label{Ch3_Thm3}
\end{theorem}
\begin{proof}
	The proof is provided in Appendix~\ref{Ch3_Sec7_SubSec4}.
\end{proof}
The pseudo code of the power allocation algorithm is presented in Algorithm~\ref{Alg3}. In the algorithm, firstly, the channel quality $\gamma_{i}$ is sorted in a descending order, as indicated in line~\ref{Ch3_Alg3:1}. The reason behind this falls into the fact that according to~\eqref{eqnoptPS}, once the optimal Lagrangian multiplier $\phi^{*}$ is determined, all links with $\gamma_{i} \leq \phi^{*}$ will be allocated to zero power and do not contribute in the calculation of the optimal Lagrangian multiplier $\phi^{*}$ in~\eqref{eqnLGRG}. With the descending $\gamma_{i}$, the corresponding Lagrangian multiplier $\phi_i$ of each link is derived as indicated in line~\ref{Ch3_Alg3:2}--\ref{Ch3_Alg3:4}. Then, the index of the optimal Lagrangian multiplier $m$ is acquired by finding the largest Lagrangian multiplier that is greater than the reciprocal of channel quality the corresponding link, as indicated in line~\ref{Ch3_Alg3:5}--\ref{Ch3_Alg3:7}. Ultimately, the optimal Lagrangian multiplier $\phi^{*}$ is determined in line~\ref{Ch3_Alg3:8} and the transmission power $p_{i}$ is allocated to each link as indicated in line~\ref{Ch3_Alg3:9}--\ref{Ch3_Alg3:11}. Similar to the slot allocation algorithm, the computational complexity of the power allocation algorithm is also $\mathcal{O}(|\mathcal{V}_\textnormal{C}|)$, as each link can only be scheduled into one group.

\begin{algorithm}[tbp]
	\renewcommand\baselinestretch{1}\selectfont
	\SetAlgoLined
	\KwIn{Channel quality $\gamma_{i}$, total transmission power $P_\textnormal{max}$, and set of links $\mathcal{M}_{s,k}$ transmitted from sender $s$ in group $k$}
	\KwOut{Power allocation policy $\bm{p}$}
	\begin{itemize}
		{\item $\phi_{i}$: Lagrangian multiplier for each link $i$}
		{\item $m$: Index of optimal Lagrangian multiplier}  
		{\item $\phi^{*}$: Optimal Lagrangian multiplier}
		{\item $p_{i}$: Allocated transmission power for link $i$}
	\end{itemize}
	\textit{Initialization}:
	$j=0$ \\
	\Begin{
		\everypar={\nl}
		Sort $\gamma_{i}$ in descending order\; \label{Ch3_Alg3:1}
		\ForEach{\textnormal{link} $i \in \mathcal{M}_{s,k}$}{ \label{Ch3_Alg3:2}
			$\frac{1}{\phi_{i}} = \frac{P_\textnormal{max} + \sum_{j=1}^{i} \frac{1}{\gamma_{j}} } {i}$\; \label{Ch3_Alg3:3}
		} \label{Ch3_Alg3:4}
		\ForEach{$\phi_{i}$}{ \label{Ch3_Alg3:5}
			Find $m$ where $\frac{1}{\phi_{m}} > \frac{1}{\gamma_{m}}$ and $\frac{1}{\phi_{m}} \leq \frac{1}{\gamma_{m+1}}$\; \label{Ch3_Alg3:6}
		} \label{Ch3_Alg3:7}
		$\frac{1}{\phi^{*}} = \frac{P_\textnormal{max} + \sum_{i=1}^{m} \frac{1}{\gamma_{i}} } {m}$\; \label{Ch3_Alg3:8}
		\ForEach{\textnormal{link} $i \in \mathcal{M}_{s,k}$}{ \label{Ch3_Alg3:9}
			$p_{i}=
			\begin{cases}
			\frac{1}{\phi^{*}} - \frac{1}{\gamma_{i}}, & \textnormal{for}\; i \in \{1,\dotsc,m\}\\
			0, & \textnormal{for}\; i \in \{m+1,\dotsc,|\mathcal{M}_{s,k}|\}\\
			\end{cases}$\; \label{Ch3_Alg3:10}
			
		} \label{Ch3_Alg3:11}
		\textbf{Return} $p_{i}$ for each link $i$\;
	}
	\caption{Power Allocation Algorithm}\label{Alg3}
\end{algorithm}

\section{Multihop Routing Algorithm Design for Path Selection}\label{Ch3_Sec4}
In the previous section, we assume that the multihop path connecting BS and UE is predefined, where the JSRA algorithm is designed based on fixed network topology. This assumption does not fully exploit the freedom given by a reconfigurable mm-wave backhaul, which can be flexibly selected as an intermediate hop for the path connecting BS and UE, and the selection depends on real-time network load, i.e., achievable data rate that varies for different paths consisting of different intermediate hops. In this section, we propose a DR algorithm for path selection, where the optimal path is greedily generated for each UE in terms of achievable data rate taking into account the data rates achieved by applying the JSRA algorithm to the existing UEs in the network. 

\subsection{Routing Algorithm Design}\label{Ch3_Sec4_SubSec1}
We introduce our DR algorithm as follows. The network is firstly established as link graph\footnote{As mentioned in Section~\ref{Ch3_Sec2_SubSec1}, UEs are associated either directly with the BS or with the geographically closest AP and connected to the BS via multihop. Depending on different network layouts specified by various use cases, this assumption can be further modified as UEs are associated with the network node (BS/AP) from which the strongest signal power is detected.} illustrated in Fig.~\ref{Ch3_Fig_Graph}. With the network topology, the achievable data rate of each link is calculated by the JSRA algorithm. Then, for all UEs in the network that requires multihop path connecting to the BS, the optimal path is greedily selected considering the existed UE(s) in the network. When a path has been determined for a UE, the network will be updated where the achievable data rate of each link is recalculated. More details of the path selection algorithm and network update are elaborated in Section~\ref{Ch3_Sec4_SubSec2} and Section~\ref{Ch3_Sec4_SubSec3}, respectively. The routing algorithm terminates when the path of all UEs are generated. A flow chart of the proposed DR algorithm is illustrated in Fig~\ref{Ch3_Fig_FC}.

\subsection{Path Selection Algorithm}\label{Ch3_Sec4_SubSec2}
For a multihop transmission between BS and UE, a source and destination node pair, $(s,d)$, is selected as either (BS,UE) or (UE,BS) for downlink or uplink transmission, respectively. Further, the achievable data rate is used as edge weight. Then, we provide the pseudo code of the proposed path selection algorithm in Algorithm~\ref{Ch3_Alg4}. 

In the beginning, both the set of traversed node $\mathcal{V}_\textnormal{t}$ and the set of next node $\mathcal{V}_\textnormal{n}$ are initialized as empty. The set of current node  $\mathcal{V}_\textnormal{c}$ contains only source node $s$, and for each node $v$ except $s$, the weight of path starting from $s$ and ending at it, referred to as $w(v)$ and defined as the minimum edge weight along the path (the achievable rate of the path), is initialized as $\infty$, which means that in the beginning all the nodes are not connected. Then, the algorithm picks every $v$ in the set of current node $\mathcal{V}_\textnormal{c}$ (starting from $s$), and finds all its neighboring node $u$ in the set of non-traversed node $\mathcal{V}\backslash \mathcal{V}_\textnormal{t}$, as described in line~\ref{Ch3_Alg4:13}--\ref{Ch3_Alg4:16}. Afterwards, there are three criteria for judging whether $u$ can be added to the path from $s$ to $v$ as the next node of $v$:
\begin{itemize}
	\item The weight of $u$ is infinity (the node has not been traversed) or adding $u$ to the path will not change the weight of path $w(u)$, as indicated in line~\ref{Ch3_Alg4:1}--\ref{Ch3_Alg4:4}. In this case, the path from $s$ to $v$ then to $u$ is one of the optimal path from $s$ to $u$ in terms of maximizing achievable data rate. 
	\item Adding the node to the path will increase the weight of path $w(u)$. In this case, the optimal path from $s$ to $u$ in terms of maximizing achievable data rate goes exclusively through $v$, as by adding the edge $(v,u)$, the minimal weight of the path is increased. Therefore, the other paths from $s$ to $u$ through other parent nodes of $u$ in the set $\mathcal{P}(u)$ should be eliminated and $v$ is the only parent node of $u$, as described in line~\ref{Ch3_Alg4:5}--\ref{Ch3_Alg4:8}.
	\item Otherwise, adding $u$ to the path may decrease the weight of path $w(u)$, which is not the desired result for maximizing the weight of the path and will not be considered.
\end{itemize}	
When the sets of currents nodes $\mathcal{V}_\textnormal{c}$ has been traversed, it will be updated as the set of next nodes, which are generated by selecting the proper neighboring nodes of nodes in $\mathcal{V}_\textnormal{c}$ by the above criteria, as described in line~\ref{Ch3_Alg4:10}--\ref{Ch3_Alg4:12}. The algorithm terminates until all nodes in the network have been assigned a path from $s$.

\begin{algorithm}[tbp]
	\renewcommand\baselinestretch{1}\selectfont
	\SetAlgoLined
	\KwIn{Link graph $\mathcal{G}(\mathcal{V},\mathcal{E})$}
	\KwOut{The optimal path $\psi^{*}$}
	\begin{itemize}
		{\item $s$: Source node}
		{\item $d$: Destination node}
		{\item $\mathcal{V}_\textnormal{t}$: Set of traversed nodes}
		{\item $\mathcal{V}_\textnormal{c}$: Set of current nodes}
		{\item $\mathcal{V}_\textnormal{n}$: Set of next nodes}
		{\item $\mathcal{P}(v)$: Set of parent nodes of node $v$}
		{\item $w(v)$: Weight of path ending at node $v$}
		{\item $w_{v\text{-}u}$: Weight of edge $(v,u)$}
	\end{itemize}
	\textit{Initialization}:
	$\mathcal{V}_\textnormal{t} = \{\}, \, \mathcal{V}_\textnormal{n} = \{\}, \, \mathcal{V}_\textnormal{c} = \{s\}, \, w(v) = \infty, \, \forall v\in \mathcal{V}\backslash s$ \\
	\Begin{
		\everypar={\nl}
		\While{$\mathcal{V}_\textnormal{t} \neq \mathcal{V}$}{ \label{Ch3_Alg4:13}
			\ForEach{$v \in \mathcal{V}_\textnormal{c}$}{ \label{Ch3_Alg4:14}
				\ForEach{$u \in \mathcal{V}\backslash \mathcal{V}_\textnormal{t}$ \textnormal{that is neighbor of} $v$}{ \label{Ch3_Alg4:16}
					\uIf{$w(u)=\infty$ \textnormal{or} $w(u)=\min\{w_{v\text{-}u}, w(v)\}$}{ \label{Ch3_Alg4:1}
						$w(u)=\min\{w_{v\text{-}u},\, w(v)\}$\; \label{Ch3_Alg4:2}
						$\mathcal{P}(u)=\mathcal{P}(u) \cup \{v\}$\; \label{Ch3_Alg4:3}
						$\mathcal{V}_\textnormal{n}=\mathcal{V}_\textnormal{n} \cup \{u\}$\; \label{Ch3_Alg4:4}  
					}
					\uElseIf{$w(u)<\min\{w_{v\text{-}u},\, w(v)\}$}{ \label{Ch3_Alg4:5}
						empty $\mathcal{P}(u)$\; \label{Ch3_Alg4:6}
						$\mathcal{P}(u) = \{v\}$\; \label{Ch3_Alg4:7}
						$\mathcal{V}_\textnormal{n}=\mathcal{V}_\textnormal{n} \cup \{u\}$\; \label{Ch3_Alg4:8}
					}
					\Else{continue\; \label{Ch3_Alg4:9}
					}
				}
				$\mathcal{V}_\textnormal{t}=\mathcal{V}_\textnormal{t} \cup \{v\}$\; \label{Ch3_Alg4:15}
			}
			empty $\mathcal{V}_\textnormal{c}$\; \label{Ch3_Alg4:10}
			$\mathcal{V}_\textnormal{c}=\mathcal{V}_\textnormal{n} $\; \label{Ch3_Alg4:11}
			empty $\mathcal{V}_\textnormal{n}$\; \label{Ch3_Alg4:12}
		}
		\textbf{Return} The optimal path $\psi^{*}: s \rightarrow  \dotsc  \rightarrow  \mathcal{P}(\mathcal{P}(d)) \rightarrow \mathcal{P}(d) \rightarrow d$\;
	}
	\caption{Path Selection Algorithm}\label{Ch3_Alg4}
\end{algorithm}

\subsection{Update Network}\label{Ch3_Sec4_SubSec3}
After the optimal path between BS and one UE is generated, the achievable data rate of each link in the network needs to be updated. This procedure utilizes the proposed JSRA algorithm, where every time the path for one UE has been determined, the achievable data rates of all links in the network are obtained by~\eqref{eqnTP}, according to the scheduling policy $\bm{\delta}$, the slot allocation policy $\bm{n}$, and the power allocation policy $\bm{p}$, which are acquired by running the JSRA algorithm for all the existing UEs in the network. In this way, the objective of DR is reached, where for each user the optimal path is selected taking into account real-time network statistics (achievable data rate) and correspondingly the system performance is improved compared to the scenario of predefined routing.

\section{Numerical Evaluation and Discussion}\label{Ch3_Sec5}
In this section, we evaluate the performance of the proposed JSRA algorithm with and without the DR algorithm for the mm-wave HetNets with both downlink and uplink traffics. The system-level evaluation setup is described in Section~\ref{Ch3_Sec5_SubSec1}. For the evaluation, we first compare the proposed algorithms with some benchmark schemes for multiplexing and interference mitigation and also with the approach that achieves theoretical optimum, in terms of data rate and latency in Section~\ref{Ch3_Sec5_SubSec2} and Section~\ref{Ch3_Sec5_SubSec3}, respectively. Then, the impacts of frame structure and duplex mode on data rate are addressed in Section~\ref{Ch3_Sec5_SubSec4} and Section~\ref{Ch3_Sec5_SubSec5}, respectively. Finally, the capability of the proposed CG-MIS scheduling algorithm in improving system performance of P2P communications, for which we focus on data delivery of a vehicular platoon, is demonstrated in Section~\ref{Ch3_Sec5_SubSec6}.  

\subsection{Simulation Setup}\label{Ch3_Sec5_SubSec1}
We consider a HetNet deployed under a single Manhattan Grid~\cite{RappaportTC, YLiWCNC, YLiTWC}, where square blocks are surrounded by streets that are 200 meters long and 30 meters wide. As illustrated in Fig.~\ref{Ch3_Fig_Simltn}, one BS and nine APs, which are marked as a black circle with black cross and black triangles, respectively, are located at the crossroads. 100 UEs are uniformly dropped in the streets marked as small blue crosses. In addition, green arrows illustrate mm-wave access links while red arrow depicts mm-wave backhaul link, respectively. The adopted propagation and antenna model are explained in Section~\ref{Ch3_Sec2_SubSec2}. It is worth noting that UEs are assumed to be almost stationary so the pathloss and shadowing values are fixed during the simulation. The type of user traffic is set as full buffer, and the default duplex mode is assumed to be half-duplex. However, for different case studies addressed in Section~\ref{Ch3_Sec5_SubSec4} and Section~\ref{Ch3_Sec5_SubSec5}, the default traffic type and duplex mode can be modified to investigate the efficiency of applying the proposed algorithms to ubiquitous system configurations. Simulation samples are averaged over 1000 independent snapshots. The default system parameter values are summarized in Table~\ref{Ch3_Table_SysMod}. 
\definecolor{myblue}{RGB}{91,155,213}
\definecolor{mygreen}{RGB}{0,204,0}
\definecolor{myred}{RGB}{255,112,112}
\definecolor{mygray}{RGB}{165,165,165}
\definecolor{myblue}{RGB}{91,155,213}
\definecolor{myyellow}{RGB}{254,192,0}
\begin{figure}[tbp]
	\centering
	\subcaptionbox{%
		\label{Ch3_Fig_FC}
		Flow chart of DR algorithm.
	}[0.48\textwidth]
	{
		\tikzstyle{decision} = [diamond, draw=white, fill=myblue!75, 
		text width=7.5em, text centered, inner sep=0pt, minimum size=7.5em, text=black]
		
		\tikzstyle{block} = [rectangle, draw=white, fill=myblue!75, 
		text width=9em, text centered, rounded corners, minimum height=5.06em, minimum width=10em, text=black]
		
		\tikzstyle{line} = [draw, very thick, -latex']
		
		\tikzstyle{startover} = [ellipse, draw=white, fill=mygray!100, text centered, text width=5.5em, minimum size=5.5em, text=white]

			\resizebox{0.4\textwidth}{!}{
				\begin{tikzpicture}[node distance = 2cm, auto]
				\tikzstyle{every node}=[font=\normalsize]
				
				\node [startover, node distance=6cm] (initNW) {Initialize Network\\$i = 1$};
				\node [decision, below of=initNW, node distance=3.75cm] (judgeUE) {$i = |\mathcal{U}| ?$};
				\node [startover, right of=judgeUE, node distance=4.5cm] (end) {End};   
				\node [block, below of=judgeUE, node distance=3.5cm] (routing) {Path selection \\ for $\textnormal{UE}_i$};
				\node [block, below of=routing, node distance=2.75cm] (updateNW) {Update network};
				\node [block, below of=updateNW, node distance=2.75cm] (i++) {$i\texttt{++}$};
				
				\path [line] (initNW) -- (judgeUE);
				\path [line] (judgeUE) -- node [midway] {Yes} (end);
				\path [line] (judgeUE) -- node [midway] {No} (routing);
				\path [line] (routing) -- (updateNW);
				\path [line] (updateNW) -- (i++);
				\draw[line] (i++.west) -- ++ (-2,0) -- ++(0,9) --  (judgeUE.west);
				
				\end{tikzpicture}} 
			}
	\subcaptionbox{%
		\label{Ch3_Fig_Simltn}
		Simulation scenario with Manhattan Grid.
		}[0.48\textwidth]{
			\resizebox{0.4\textwidth}{!}{ 
				\begin{tikzpicture}
				\tikzset{cross/.style={cross out, draw=black, minimum size=2*(#1-\pgflinewidth), inner sep=0pt, outer sep=0pt},cross/.default={2pt}};
				\tikzset{UE/.style={cross out, draw=myblue, minimum size=2*(#1-\pgflinewidth), inner sep=0pt, outer sep=0pt, rotate=45},UE/.default={1.5pt}};
				
				\draw[color=black] (-0.069,-0.06) -- (0.069,-0.06) -- (0,0.06)-- (-0.069,-0.06);
				\draw[color=black] (-0.069+1,-0.06) -- (0.069+1,-0.06) -- (0+1,0.06)-- (-0.069+1,-0.06);
				\draw[color=black] (-0.069+2,-0.06) -- (0.069+2,-0.06) -- (0+2,0.06)-- (-0.069+2,-0.06);
				\draw[color=black] (-0.069,-0.06+1) -- (0.069,-0.06+1) -- (0,0.06+1)-- (-0.069,-0.06+1);
				\draw[color=black] (-0.069+2,-0.06+1) -- (0.069+2,-0.06+1) -- (0+2,0.06+1)-- (-0.069+2,-0.06+1);
				\draw[color=black] (-0.069,-0.06+2) -- (0.069,-0.06+2) -- (0,0.06+2)-- (-0.069,-0.06+2);
				\draw[color=black] (-0.069+1,-0.06+2) -- (0.069+1,-0.06+2) -- (0+1,0.06+2)-- (-0.069+1,-0.06+2);
				\draw[color=black] (-0.069+2,-0.06+2) -- (0.069+2,-0.06+2) -- (0+2,0.06+2)-- (-0.069+2,-0.06+2);
				\draw[color=black] (1, 1) circle (0.075);
				
				\draw (-0.1,-0.1) -- (-0.1,2.1) -- (2.1,2.1) -- (2.1,-0.1) -- cycle;
				
				\foreach \x in {0.1,1.1}{
					\foreach \y in {0.1,1.1}{
						\draw (\x,\y) -- (\x+0.8,\y) -- (\x+0.8,\y+0.8) -- (\x,\y+0.8) -- cycle;
					}
				}
				
				\draw (1,1) node[cross]{};
				
				\draw(0,0.67)node[UE]{};
				\draw(0.6,1.04)node[UE]{};
				\draw(1.95,1.47)node[UE]{};
				\draw(1.5,0)node[UE]{};
				\draw(1.33,2)node[UE]{};
				\draw(1.43,0.96)node[UE]{};
				\draw(1.05,0.675)node[UE]{};
				\draw(0.05,1.25)node[UE]{};
				\draw(0.6,1.95)node[UE]{};
				\draw(0.7,0)node[UE]{};
				\draw(2.03,0.7)node[UE]{};
				\draw [-latex',draw=mygreen] (0.925,1) -- node[right] {} (0.67,1.04);
				\draw [-latex',draw=mygreen] (1.26,2) -- node[right] {} (1.075,2);
				\draw [-latex',draw=myred] (1,1.925) -- node[right] {} (1,1.075);
				
				\end{tikzpicture}} 
			}
	\caption{Illustration of (a) flow chart of DR algorithm and (b) simulation scenario with Manhattan Grid.}\label{Ch3_Fig_DR-Simltn}
\end{figure}
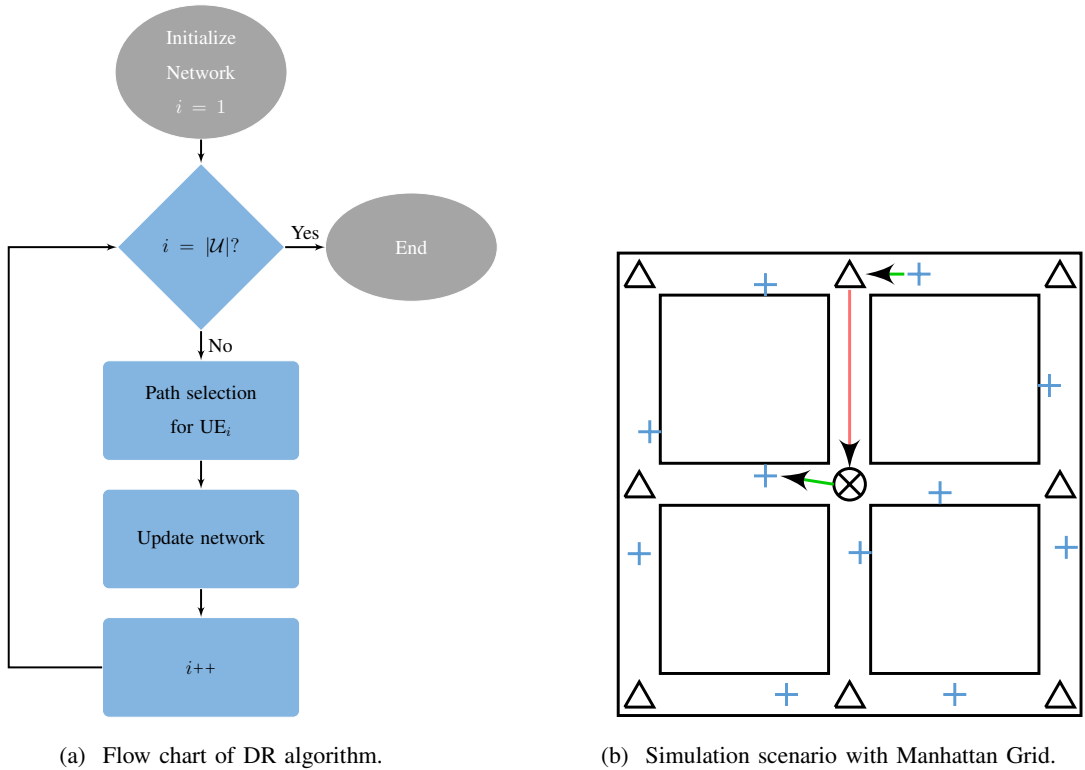

\subsection{Case Studies: Data Rate}\label{Ch3_Sec5_SubSec2}
In Fig.~\ref{Ch3_Fig_DR:a}, we plot the simulation results of the data rate for different schemes in the considered HetNet, versus the type of data rate. Here, the edge data rate is defined as the 5-th percentile point of the cumulative distribution function (CDF) of data rates. In particular, TDMA and eICIC~\cite{Lopez-Perez} schemes are selected as benchmark schemes for the performance comparison with the proposed algorithms. As shown in Fig.~\ref{Ch3_Fig_DR:a}, the proposed JSRA algorithm provides considerable improvement in both edge and average data rate, thanks to the spatial multiplexing gain achieved by enabling simultaneous transmission, compared to the benchmark schemes. Moreover, applying the DR algorithm to the JSRA algorithm taking into account real-time network statistics further improves the data rate compared to the JSRA algorithm with predefined routing. In particular, the proposed JSRA algorithm with the DR algorithm closely approximates to the optimal data rate, which is achieved by a general dynamic policy approach (DPP) for the network utility maximization problem~\cite{Neely}, where at each time slot a maximum weighted sum rate problem with respect to single-hop instantaneous rate is solved, and weights are recursively updated in terms of the flow queues at each node.
\begin{figure}[tbp]
	\centering
	\begin{subfigure}{0.48\textwidth}
		\includegraphics[width=\textwidth]{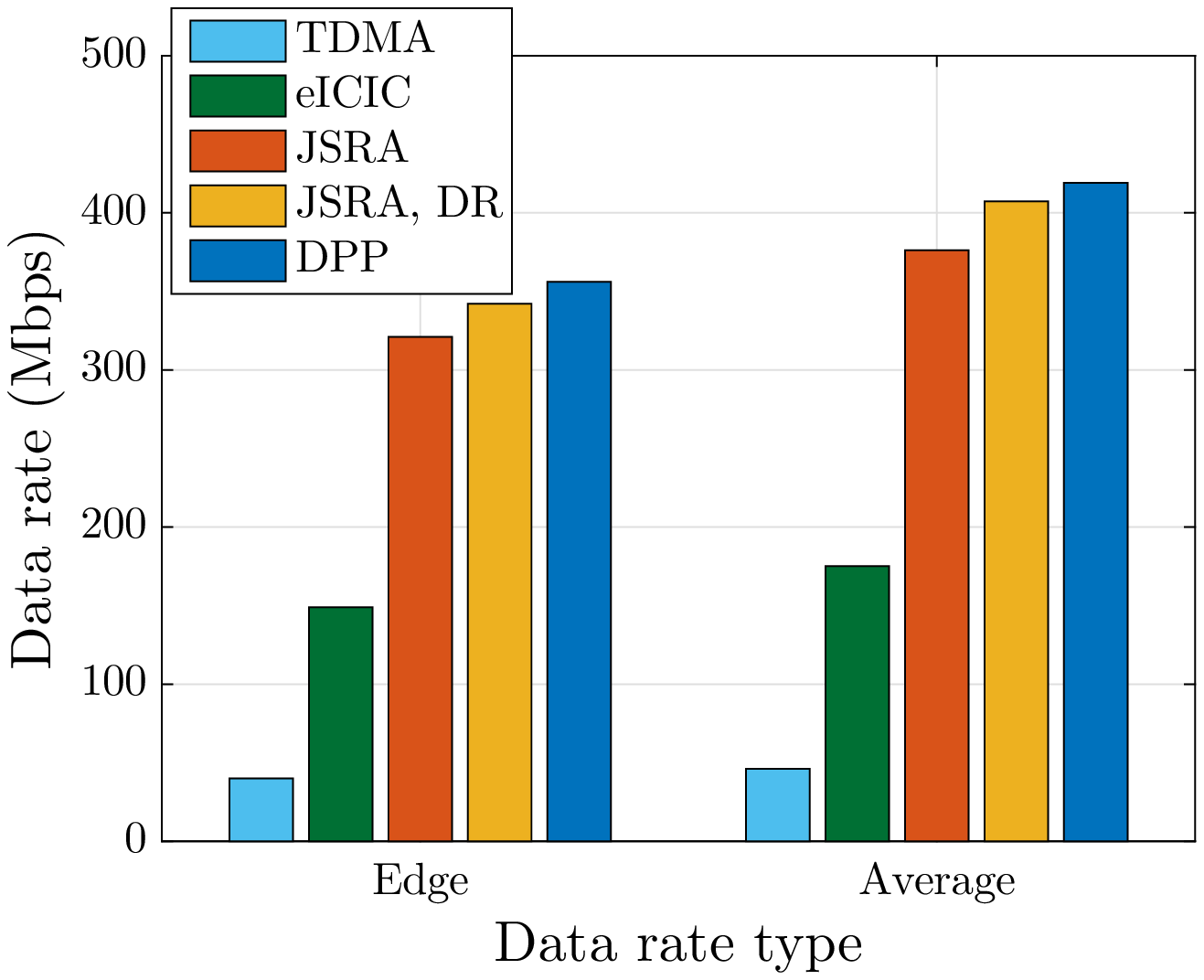}
		\caption{Data rate type}
		\label{Ch3_Fig_DR:a}
	\end{subfigure}
	\begin{subfigure}{0.48\textwidth}
		\includegraphics[width=\textwidth]{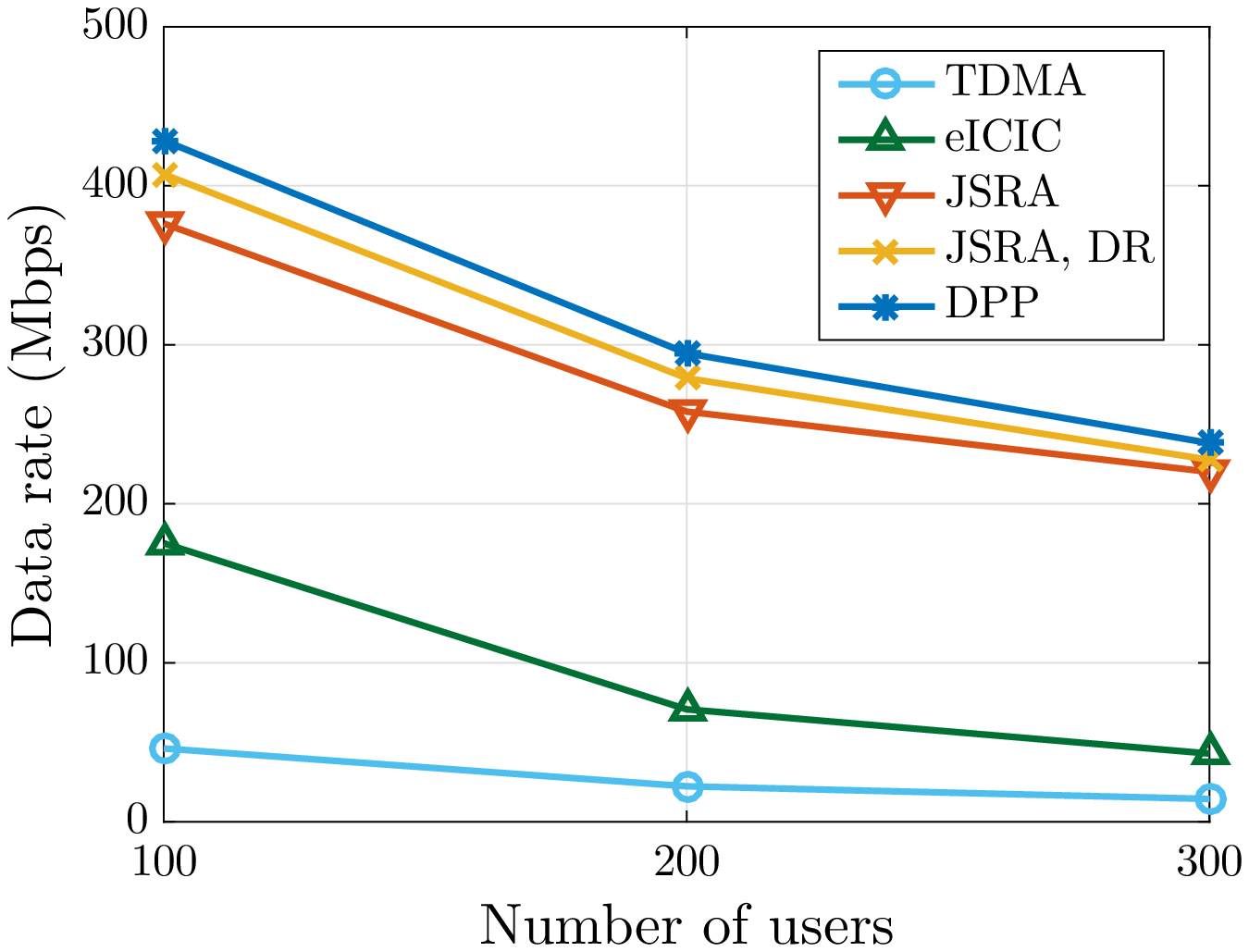}
		\caption{Number of users}
		\label{Ch3_Fig_DR:b}
	\end{subfigure}
	\caption{Performance of data rate for different schemes are compared in different (a) data rate type and (b) number of users.}\label{Ch3_Fig_DR}
\end{figure}

In Fig.~\ref{Ch3_Fig_DR:b}, we plot the simulation results of the data rate for different schemes in the considered HetNet, versus the number of users in the network. On the one hand, as expected, increasing the number of users reduces the average data rate due to limited resource. However, the proposed JSRA algorithm (without and with the DR algorithm) still provides significant improvement compared to the benchmark schemes for different number of users. On the other hand, with the increased user density, the gap between the optimum (DPP) and the proposed algorithms shrinks. The reason behind this falls into the fact that when the number of users grows, the allocatable resource to each link in all schemes is limited and becomes the dominant factor in determining the data rate.

\subsection{Case Studies: Latency}\label{Ch3_Sec5_SubSec3}
While the DPP scheme is (asymptotically) optimal for maximizing the data rate and outperforms the proposed JSRA algorithm (without and with the DR algorithm), due to the gain from slot-wise scheduling and resource allocation compared to the frame-wise approach of the JSRA algorithm, it is not optimal when we wish to guarantee other system performance, e.g.\ latency. In other words, DPP scheme achieves high data rate by sacrificing some flow, which might never be scheduled (always stay in flow queue), to maintain the stability of each transmission node in terms of the balance between influx and outflow bits. 
\begin{figure}[tbp]
	\centering
	\begin{subfigure}{0.48\textwidth}
		\includegraphics[width=\textwidth]{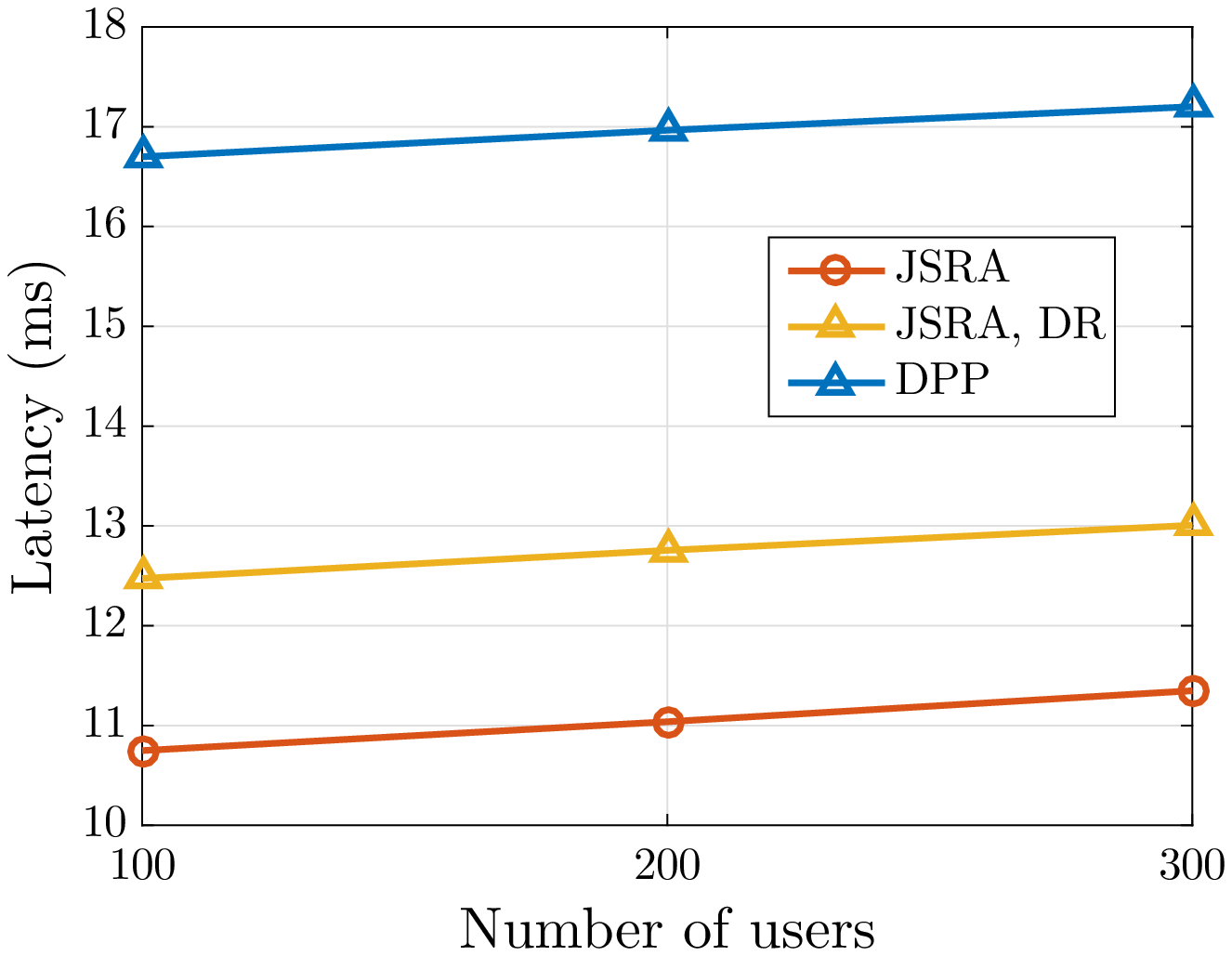}
		\caption{Latency}
		\label{Ch3_Fig_LFS:a}
	\end{subfigure}
	\begin{subfigure}{0.48\textwidth}
		\includegraphics[width=\textwidth]{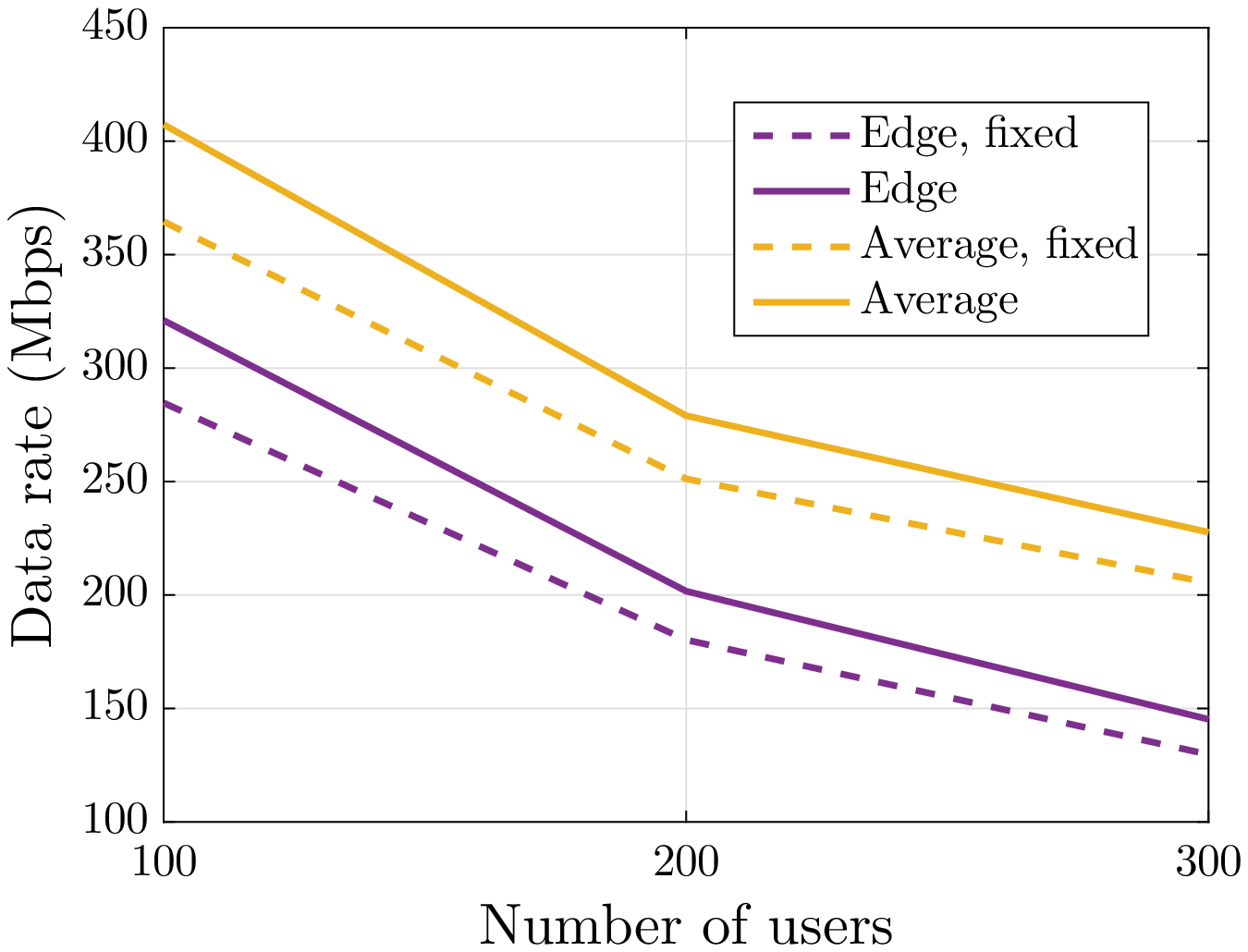}
		\caption{Data rate}
		\label{Ch3_Fig_LFS:b}
	\end{subfigure}
	\caption{Performance of latency for different schemes are compared in different number of users in (a). Performance of data rate for different frame structures are compared in different number of users in (b).}\label{Ch3_Fig_LFS}
\end{figure}

In Fig.~\ref{Ch3_Fig_LFS:a}, we plot the simulation results of the latency for different schemes in the considered HetNet, versus the number of users in the network. As stated above, DPP scheme experiences higher latency compared to the JSRA algorithm (without and with the DR algorithm). According to the path selection criteria of the DR algorithm, a link that achieves higher data rate is more likely to be picked as an intermediate hop for a path, which may lead to the situation that the optimal path in terms of maximizing achievable data rate consists of a large number of hops which eventually cause higher latency compared to fixed routing with limited number of hops. 

\subsection{Case Studies: Frame Structure}\label{Ch3_Sec5_SubSec4}
As the actual user demand, which is represented by the required transmission slot, varies across different frames, the results of the slot allocation policy $\bm{n}$ for different frames may also be diverse. This brings an additional advantage for frame structure design, where the ``switch point'' for different node (BS/AP/UE), defined as the percentage of the number of downlink slots in each frame, can be flexibly adjusted according to actual allocated slots to downlink and uplink transmissions. 

In Fig.~\ref{Ch3_Fig_SP}, we plot the simulation results of switch points for four APs in the considered HetNet, versus the index of successive frames. The trend of switch point curves for the other BS/AP are expected to be similar to the depicted ones. As anticipated, the switch points of the selected APs fluctuate in the vicinity of the baseline ($50\%$, which means the number of downlink and uplink slots are the same) for both algorithms, which shows the efficiency the proposed algorithm in allocating slots to fulfill actual user demands.
\begin{figure}[tbp]
	\centering
	\includegraphics[width=0.6\textwidth]{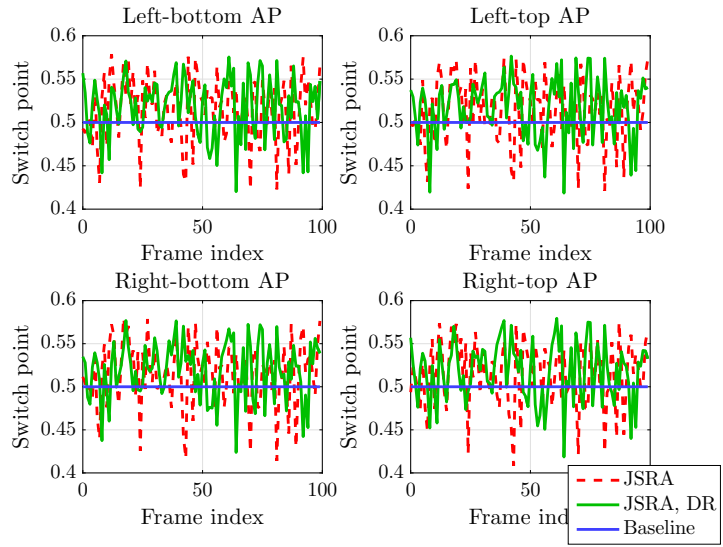}
	\caption{Switch point for different APs are illustrated across different frames.}\label{Ch3_Fig_SP}
\end{figure}

In Fig.~\ref{Ch3_Fig_LFS:b}, we plot the simulation results of the data rate for JSRA algorithm with frame structure of flexible switch point and with frame structure of fixed switch point ($50\%$) in the considered HetNet, versus the number of users in the network. Here, the scheduling policy $\bm{\delta}$ acquired by the CG-MIS algorithm for the above two cases are the same, but the slot allocation in the fixed switch point case is restricted to allocating total number of slots in each frame equally to all downlink and uplink links (namely $50\%$ -- $50\%$). In the figure, we notice that the proposed algorithm, owing to the capability of the flexible adjustment of downlink/uplink slots, yields both higher edge and average data rates compared to that of fixed switch point.

\subsection{Case Studies: Duplex Mode}\label{Ch3_Sec5_SubSec5}
In addition to the default half-duplex mode, in Fig.~\ref{Ch3_Fig_FullD:a} and Fig.~\ref{Ch3_Fig_FullD:b}, we plot the simulation results of the data rate for different duplex modes in the considered HetNet, versus the type of data rate and the number of users, respectively. Here, the perfect full-duplex refers to the case that a reception link is isolated from the simultaneously scheduled transmission link of the same node without interference, while the interfered (abbreviated as interf.\ in the figures) case incorporates a $\SI{-110}{\decibel}$ interference at the reception link caused by the simultaneously scheduled transmission link at the same node. The full-duplex modes are further classified into only enabling full-duplex mode at AP (green and red bars in Fig.~\ref{Ch3_Fig_FullD:a} and Fig.~\ref{Ch3_Fig_FullD:b}) and at both AP and BS (yellow and blue bars). Results show that by relaxing the scheduling criterion for simultaneous transmission (from half-duplex to full-duplex, from full-duplex at only AP to at both AP and BS, etc.), the data rate is gradually improved.
\begin{figure}[tbp]
	\centering
	\begin{subfigure}{0.48\textwidth}
		\includegraphics[width=\textwidth]{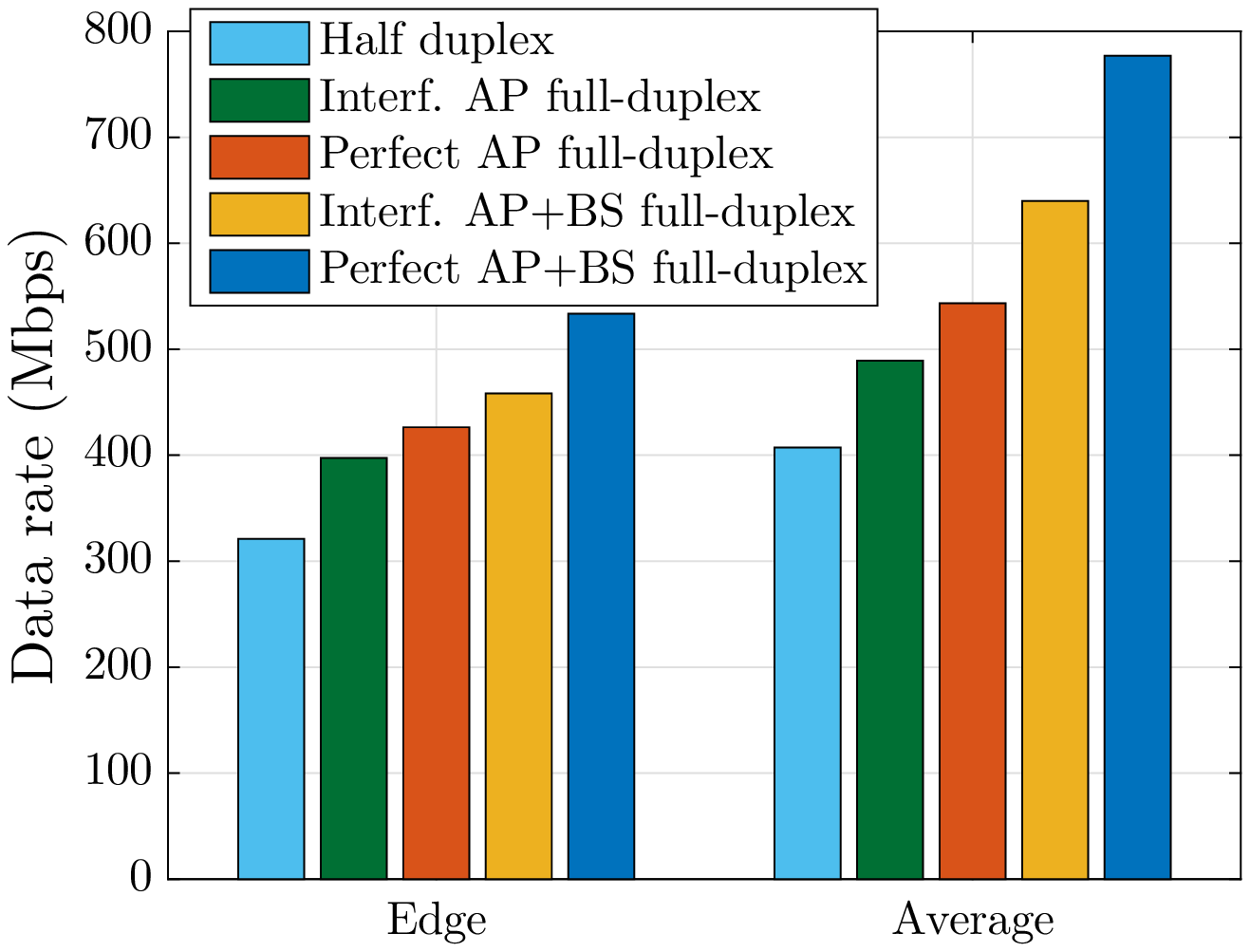}
		\caption{Data rate type}
		\label{Ch3_Fig_FullD:a}
	\end{subfigure}
	\begin{subfigure}{0.48\textwidth}
		\includegraphics[width=\textwidth]{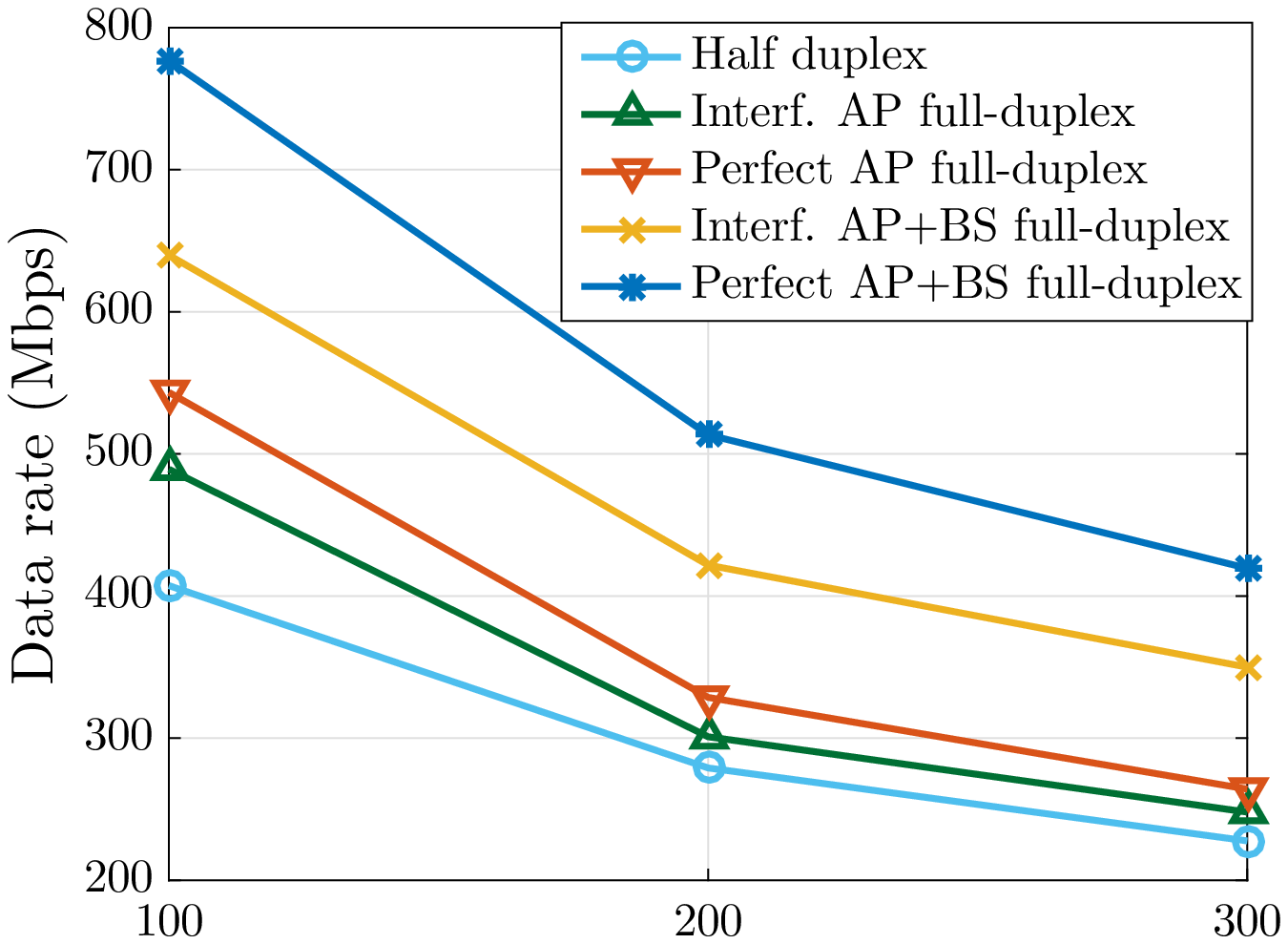}
		\caption{Number of users}
		\label{Ch3_Fig_FullD:b}
	\end{subfigure}
	\caption{Performance of data rate for different duplex modes are compared in different (a) data rate type and (b) number of users.}\label{Ch3_Fig_FullD}
\end{figure}

\subsection{Case Studies: P2P Communications}\label{Ch3_Sec5_SubSec6}
In this subsection, we focus on the capability of the proposed JSRA algorithm in P2P communication scenario. Specifically, we consider the data delivery in a platoon of vehicles, which is currently intensively studied by standardization association and research activities~\cite{3GPPV2X}. It is obvious that for vehicles in a platoon, P2P communications are more likely to be established, as vehicles move in a line and data is forwarded by the vehicles one after another. When the number of vehicles in a platoon increases, the delivery latency of data becomes a crucial target of performance optimization, for which we may concentrate on how to schedule the data streaming in a platoon. Moreover, as bumper antennas installed in vehicles are geographically separated and non-collocated, referring to the method of enabling simultaneous transmission and reception as full-duplex is no longer appropriate. Hence, the method in vehicular society has been proposed as \textit{bidirectional transmission} (BT), which is designed to be differentiated from conventional full-duplex. 

Based on the above observations, in this subsection we study how the system performance of data delivery in a platoon benefits from the proposed CG-MIS scheduling algorithm. We focus on a platoon of ten vehicles in the considered HetNet, where each vehicle has data to transmit to other vehicles in the platoon or to the BS, or vice versa. The vehicle speed is assumed to be $\SI{50}{\kilo\meter/\hour}$, and the distance between two adjacent vehicles equals to vehicle speed (in $\si{\meter/\second}$)$\, \times 2$. 
\begin{figure}[tbp]
	\centering
	\begin{subfigure}{0.48\textwidth}
		\includegraphics[width=\textwidth]{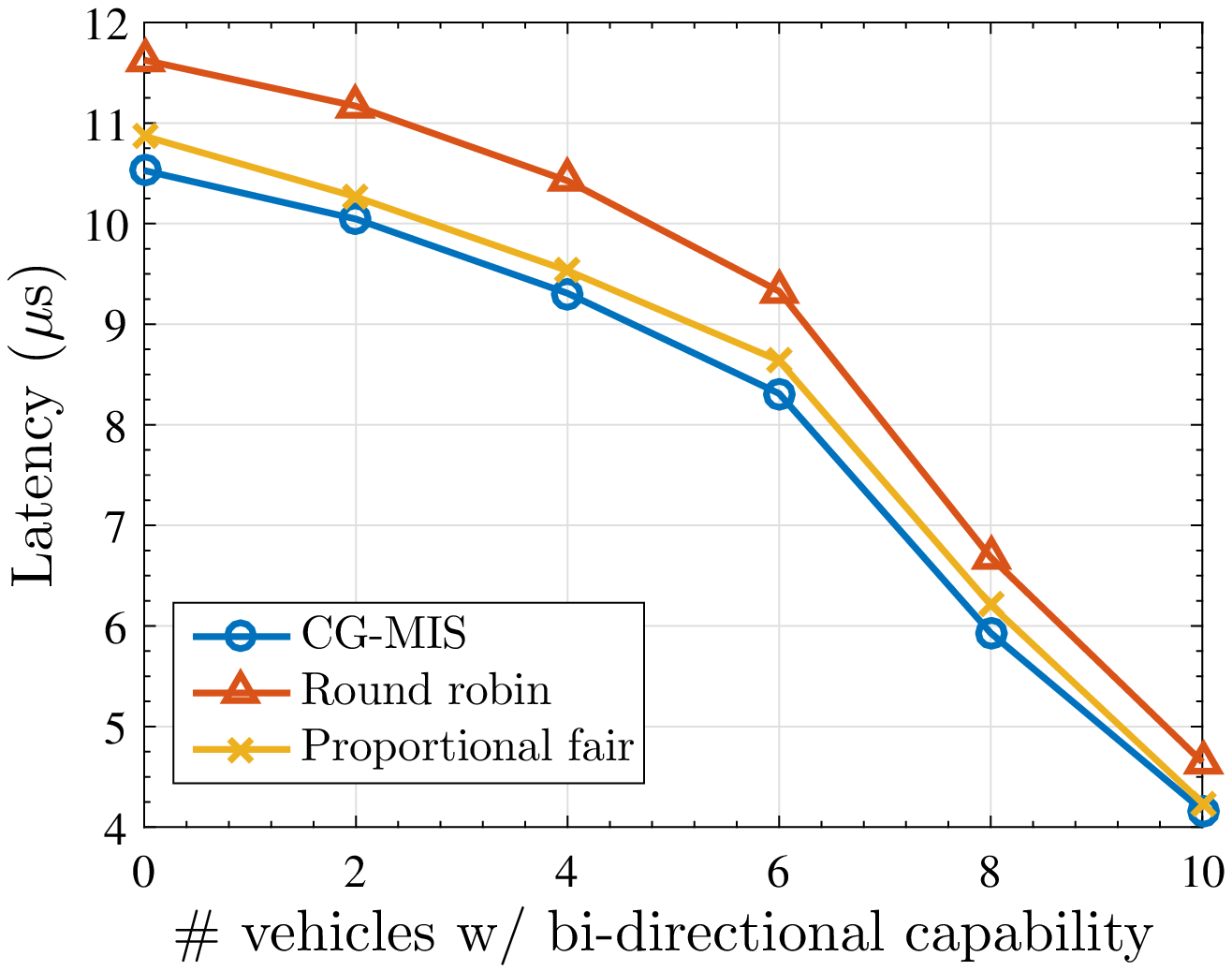}
		\caption{Data rate type}
		\label{Ch3_Fig_V:a}
	\end{subfigure}
	\begin{subfigure}{0.48\textwidth}
		\includegraphics[width=\textwidth]{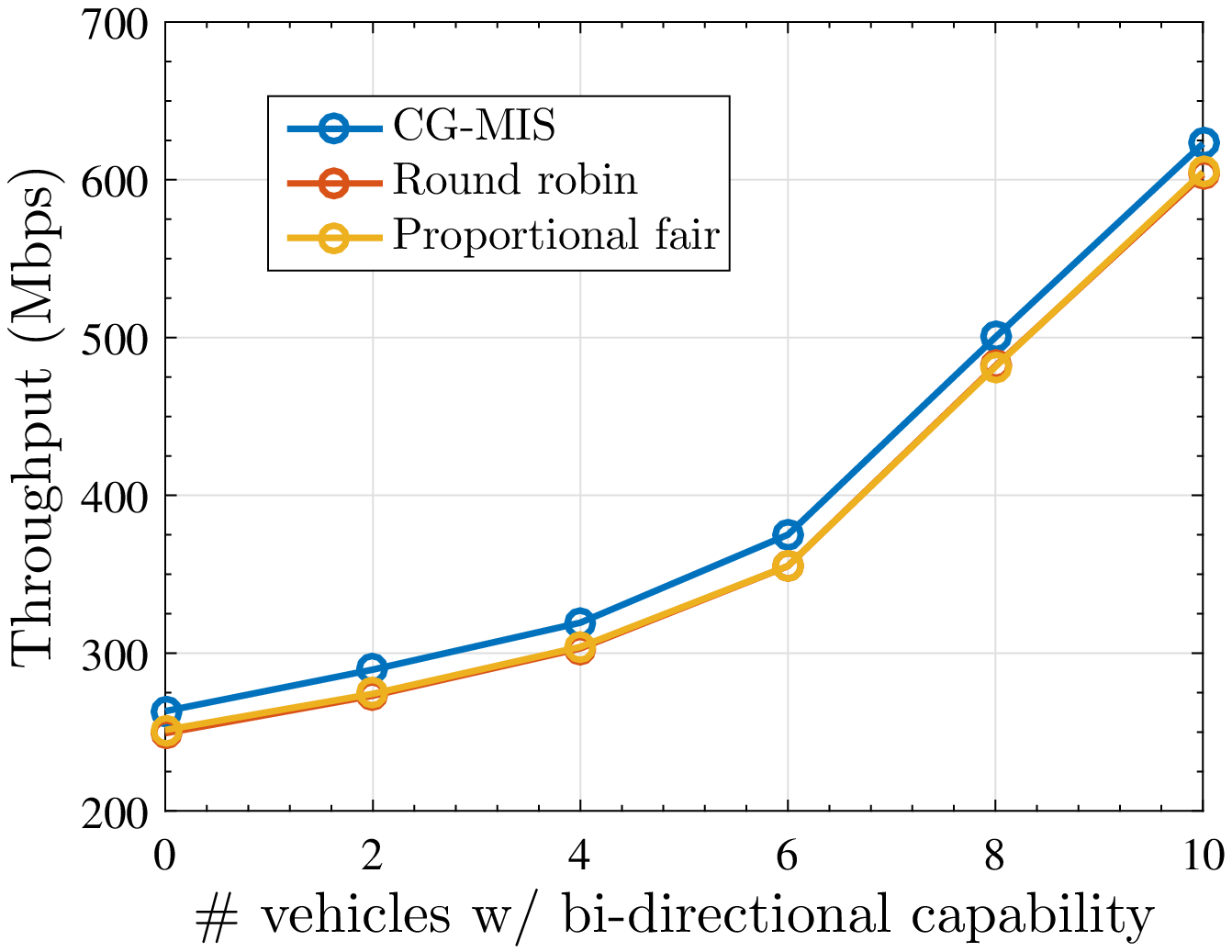}
		\caption{Number of users}
		\label{Ch3_Fig_V:b}
	\end{subfigure}
	\caption{Performance for different scheduling schemes with different number of BT-enabled vehicles are compared in (a) latency and (b) throughput.}\label{Ch3_Fig_V}
\end{figure}

In Fig.~\ref{Ch3_Fig_V:a} and Fig.~\ref{Ch3_Fig_V:b}, we plot the simulation results of the latency and the throughput for different scheduling schemes in the platoon, versus the number of BT-enabled vehicles, respectively. On the one hand, the results suggest that the more BT-enabled vehicles, the lower the latency and the higher the throughput are. Here, for BT-enabled vehicles, transmission and reception links are simultaneous scheduled, which allows data to ``flow'' through these vehicles and eventually decreases the latency. On the other hand, the proposed CG-MIS scheduling algorithm outperforms other classical scheduling algorithms, namely round robin and proportional fair algorithm, in both latency and throughput. The reason behind this lies in the fact that the CG-MIS scheduling algorithm groups links that are suitable to transmit simultaneously according to achievable channel capacity, which is better than the round robin algorithm that always transmits the first data flow in a queue, and the proportional fair algorithm that delays links with low achievable channel capacity due to interference, respectively.

\section{Summary}\label{Ch3_Sec6}
In this paper, we have addressed the problem of maximizing the achievable data rate of mm-wave multihop HetNets considering both downlink and uplink transmissions on backhaul and access links with IAB structure. To solve the maximization problem in an efficient way, we proposed a joint scheduling and resource allocation algorithm, where the optimization problem was decomposed into subproblems including link scheduling, transmission slot allocation, and transmission power allocation. The subproblems were then solved by the proposed maximum independent set based scheduling algorithm, the proportional fair slot allocation algorithm, and the water-filling power allocation algorithm, respectively. Based on this, a dynamic routing algorithm, which incorporates a path selection algorithm to greedily search the optimal path connecting BS and UE by considering real-time network load and traffic, was investigated to further improve the data rate. By evaluating the proposed algorithms via extensive simulations, we concluded that the proposed joint scheduling and resource allocation algorithm, together with the dynamic routing algorithm, outperforms benchmark schemes in terms of achievable data rate, and closely approach the theoretical optimum, yet with lower latency. Besides, the proposed algorithms achieve higher data rate, by enabling the flexible adjustment of downlink and uplink slot allocation according to spontaneous user demand, and support both half- and full-duplex modes with considerable performance enhancement. In particular, the proposed algorithms can deliver significant flexibility in terms of fulfilling various performance requirements for both point-to-point and point-to-multipoint communications.

Future work can consider a joint optimization of link scheduling, resource allocation scheme, and also routing strategy to achieve even better overall system performance compared to the current solutions, while a potential increase in computational complexity should be also investigated. It would also be interesting to leverage the proposed solution to multi-connectivity scenarios where users are allowed to connect with multiple BSs, and/or to other interference models.

\appendices

\section{Proof of Theorem~\ref{Ch3_Thm3}}\label{Ch3_Sec7_SubSec4}
As the transmission power of a link is non-negative, the relaxed optimization problem~\ref{Ch3_prb3} provided by~\eqref{eqnoptPrlx} needs to be modified as
\begin{align*}
\underset{\bm{p}}\max &\null \quad \sum_{i \in \bm{M}_{s,k}} \log \left(1+\gamma_{i} p_{i}\right), \\
\textnormal{s.t.} &\null \quad \sum_{i \in \mathcal{M}_{s,k}} p_{i} \leq P_\text{max}, \; p_{i} \geq 0, \; \forall i \in \mathcal{M}_{s,k},  \numberthis \label{eqnoptPrlxpf1}
\end{align*}
where $\gamma_{i} = \frac{g_i l_i^{-1}}{\eta}$. For simplicity, we assume $\log = \ln$. Observe that all inequality constraints functions in~\eqref{eqnoptPrlxpf1} are affine, and at least a set of $\{p_{i}\}$ exists such that $p_{i} = \frac{P_\textnormal{max}} {(|\mathcal{M}_{s,k}|+1)} > 0$, and $\sum_{i \in \mathcal{M}_{s,k}} p_{i} \leq P_\text{max}$, which implies the Karush-Kuhn-Tucker (KKT) conditions are necessary. Moreover, the objective function in~\eqref{eqnoptPrlxpf1} is the sum of convex functions and consequently convex as well. Therefore, the KKT conditions are also sufficient, in which we can use standard KKT form to solve the problem. They are concluded as follows:
\begin{itemize}
	\item Primal Feasibility (PF):
	\begin{align}
	\sum_{i \in \mathcal{M}_{s,k}} p_{i} \leq P_\text{max}, \;p_{i} \geq 0, \; \forall i \in \mathcal{M}_{s,k}. \label{eqnoptPrlxpf2}
	\end{align}
\end{itemize}
\begin{itemize}
	\item Dual Feasibility (DF):
	\begin{align}
	&-\frac{\gamma_{i}} {1+\gamma_{i}p_{i}} + \phi - \omega_{i} = 0, \, \forall i \in \mathcal{M}_{s,k}, \; \phi \geq 0, \; \omega_{i} \geq 0,\; \forall i \in \mathcal{M}_{s,k}. \label{eqnoptPrlxpf3}
	\end{align}
\end{itemize}
\begin{itemize}
	\item Complementary Slackness (CS):
	\begin{align}
	{\phi} {\left(\sum_{i \in \mathcal{M}_{s,k}} {p_{i} - P_\textnormal{max}}\right)} = 0, \; \omega_{i}p_{i} = 0,\; \forall i \in \mathcal{M}_{s,k}. \label{eqnoptPrlxpf4} 
	\end{align}
\end{itemize}

As the channel capacity $r_{i}$ is increasing with $p_{i}$, the optimal power allocation for link $i$, denoted as $p_{i}^{*}$, should satisfy $\sum_{i \in \mathcal{M}_{s,k}} p_{i}^{*} = P_\textnormal{max}$. As a result, we can always increase some $p_{i}$ without violating the first PF condition to increase the sum rate, in case $\sum_{i \in \mathcal{M}_{s,k}} p_{i} < P_\textnormal{max}$, which means the first CS condition is always satisfied and there is no further restriction on the Lagrangian multiplier $\phi$ besides the second DF condition. Actually,~\eqref{eqnoptPrlxpf3} also indicates that $\phi > 0$. Otherwise, assume $\phi = 0$, then the other two DF conditions require $0 > -\frac{\gamma_{i}} {1+\gamma_{i}p_{i}} - \omega_{i} = 0$, and $\omega_{i} \geq 0$, which results in contradiction.

Now, let link $j$ gets positive transmission power, i.e., $p_{j} > 0$. Then, according to~\eqref{eqnoptPrlxpf3} and~\eqref{eqnoptPrlxpf4}, we have $\omega_{j} = 0$ and $\frac{\gamma_{j}} {1+\gamma_{j}p_{j}} = \phi$, which means $p_{j} = \frac{1}{\phi} - \frac{1}{\gamma_{j}}$. Note that this can hold only if $\frac{1}{\phi} - \frac{1}{\gamma_{j}} > 0$, or equivalently, $\gamma_{j} > \phi$. Otherwise, to satisfy the first DF condition, we must have $p_{j} = 0$, and $\omega_{j} = \phi - \gamma_{j}$.

In conclusion, the optimal solution of problem~\ref{Ch3_prb3} is given by
\begin{equation}
p_{i}^{*} = \max\left\{ \frac{1}{\phi^{*}} - \frac{1}{\gamma_{i}}, 0\right\},
\label{eqnoptPrlxpf22}
\end{equation}
where the optimal Lagrangian multiplier $\phi^{*}$ can be derived from 
\begin{equation}
\sum\limits_{i \in \mathcal{M}_{s,k}}{\max\left\{ \frac{1}{\phi^{*}} - \frac{1}{\gamma_{i}}, 0\right\}} = P_\textnormal{max}.
\label{eqnoptPrlxpf23}
\end{equation}

\bibliographystyle{IEEEtran}
\bibliography{mm_wave}

\end{document}